\documentclass[11pt]{article}

\usepackage{graphicx,amssymb,amsmath}
\usepackage{amsfonts}
\usepackage{booktabs}
\usepackage{cite}
\usepackage{color}
\usepackage{enumerate}
\usepackage{float}
\usepackage{hyperref}
\usepackage{mathrsfs}
\usepackage{subfig}
\usepackage{tabularx}
\usepackage{authblk}

\usepackage[in]{fullpage}
\usepackage[top=0.8in, bottom=0.9in, left=0.9in, right=0.9in]{geometry}

\newcommand{\fvd}{\mbox{$F\!V\!D$}}
\newcommand{\gch}{\mbox{$G\!C\!H$}}

\def\calL{\mathcal{L}}
\newtheorem{observation}{Observation}

\newtheorem{lemma}{Lemma}
\newtheorem{theorem}{Theorem}
\newenvironment{proof}{\noindent {\textbf{Proof:}}\rm}{\hfill $\Box$ \rm\bigskip}

\title{An Optimal Deterministic Algorithm for Geodesic Farthest-Point Voronoi Diagrams in Simple Polygons\thanks{This research was supported in part by NSF under Grant CCF-2005323. A preliminary version of this paper will appear in Proceedings of the 37th International Symposium on Computational Geometry (SoCG 2021).}}

\author{
Haitao Wang
}
\affil{Department of Computer Science \\
Utah State University, Logan, UT 84322, USA
\\ {\tt haitao.wang@usu.edu}}

\begin{document}
\pagestyle{plain}
\date{}

\thispagestyle{empty}
\maketitle

\vspace{-0.25in}
\begin{abstract}
Given a set $S$ of $m$ point sites in a simple polygon $P$ of $n$ vertices, we consider the problem of computing the geodesic farthest-point Voronoi diagram for $S$ in $P$. It is known that the problem has an $\Omega(n+m\log m)$ time lower bound. Previously, a randomized algorithm was proposed [Barba, SoCG 2019] that can solve the problem in $O(n+m\log m)$ expected time. The previous best deterministic algorithms solve the problem in $O(n\log \log n+ m\log m)$ time [Oh, Barba, and Ahn, SoCG 2016] or in $O(n+m\log m+m\log^2 n)$ time [Oh and Ahn, SoCG 2017]. In this paper, we present a deterministic algorithm of $O(n+m\log m)$ time, which is optimal. This answers an open question posed by Mitchell in the Handbook of Computational Geometry two decades ago.
\end{abstract}

\section{Introduction}
\label{sec:intro}

Let $P$ be a simple polygon of $n$ vertices in the plane. Let $S$ be a set of $m$ points, called {\em sites}, in $P$ (each site can be either in the interior or on the boundary of $P$). For any two points in $P$, their {\em geodesic distance} is the length of their Euclidean shortest path in $P$. We consider the problem of computing the geodesic farthest-point Voronoi diagram of $S$ in $P$, which is to partition $P$ into Voronoi cells such that all points in the same cell have the same farthest site in $S$ with respect to the geodesic distance.

This problem generalizes the Euclidean farthest Voronoi diagram of $m$ sites in the plane, which can be computed in $O(m\log m)$ time~\cite{ref:PreparataCo85}; this is optimal as $\Omega(m\log m)$ is a lower bound. For the more general geodesic problem in $P$, Aronov et al.~\cite{ref:AronovTh93} showed that the complexity of the diagram is $\Theta(n+m)$ and provided an $O(n\log n+m\log m)$ time algorithm. The runtime is close to optimal as $\Omega(n+m\log m)$ is a lower bound. No progress had been made for over two decades until in SoCG 2016 Oh et al.~\cite{ref:OhTh20} proposed an $O(n\log\log n+m\log m)$ time algorithm. Later in SoCG 2017 Oh and Ahn~\cite{ref:OhVo20} gave another $O(n+m\log m+m\log^2 n)$ time algorithm and in SoCG 2019 Barba~\cite{ref:BarbaOp19} presented a randomized algorithm that can solve the problem in $O(n+m\log m)$ expected time.

In this paper, we give an $O(n+m\log m)$ time deterministic algorithm, which is optimal. The space complexity of the algorithm is $O(n+m)$. This answers an open question posed by Mitchell~\cite{ref:MitchellGe00} in the Handbook of Computational Geometry two decades ago.

\subsection{Related work}

If all sites of $S$ are on the boundary of $P$, then better results exist. The algorithm of Oh et al.~\cite{ref:OhTh20} can solve the problem in $O((n+m)\log\log n)$ time while the randomized algorithm of Barba~\cite{ref:BarbaOp19} runs in $O(n+m)$ expected time.

The geodesic nearest-point Voronoi diagram for sites in a
simple polygon has also attracted much attention. The problem also has an
$\Omega(n+m\log m)$ time lower bound. The first close-to-optimal
algorithm was given by Aronov~\cite{ref:AronovOn89} and the running
time is $O((n+m)\log(n+m)\log n)$. Papadopoulou
and Lee~\cite{ref:PapadopoulouA98} improved the algorithm to
$O((n+m)\log (n+m))$ time. Recent progress has been made by Oh and
Ahn~\cite{ref:OhVo20} who presented an $O(n+m\log m\log^2 n)$ time
algorithm and also by Liu~\cite{ref:LiuA20} who designed an
$O(n+m(\log m + \log^2 n))$ time algorithm. Finally the problem was
solved optimally in $O(n+m\log m)$ time by Oh~\cite{ref:OhOp19}.

Another closely related problem is to compute the {\em geodesic
center} of a simple polygon $P$, which is a point in $P$ that
minimizes the maximum geodesic distance from all points of $P$. Asano
and Toussaint~\cite{ref:AsanoCo85} first gave an $O(n^4\log n)$ time
algorithm for the problem. Pollack et al.~\cite{ref:PollackCo89}
derived an $O(n\log n)$ time algorithm. Recently the problem was
solved optimally in $O(n)$ time by Ahn et al.~\cite{ref:AhnA16}. The
{\em geodesic diameter} of $P$ is the largest geodesic distance between any two points in $P$. Chazelle~\cite{ref:ChazelleA82} first gave an $O(n^2)$ time algorithm and then Suri~\cite{ref:SuriCo89} presented an improved $O(n\log n)$ time solution. Hershberger and Suri~\cite{ref:HershbergerMa97} finally solved the problem in $O(n)$ time.

All above results are for simple polygons. For polygons with holes, the problems become more difficult. The geodesic nearest-point Voronoi diagram for $m$ point sites in a polygon with holes of $n$ vertices can be solved in $O((n+m)\log (n+m))$ time by the algorithm of Hershberber and Suri~\cite{ref:HershbergerAn99}. Bae and Chwa~\cite{ref:BaeTh09} gave an algorithm for constructing the geodesic farthest-point Voronoi diagram and the algorithm runs in $O(nm \log^2(n+m) \log m)$ time. For computing the geodesic diameter in a polygon with holes, Bae et al.~\cite{ref:BaeTh13} solved the problem in $O(n^{7.73})$ or $O(n^7(h+\log n))$ time, where $h$ is the number of holes. For computing the geodesic center, Bae et al.~\cite{ref:BaeCo19} first gave an $O(n^{12+\epsilon})$ time algorithm, for any constant $\epsilon>0$; Wang~\cite{ref:WangOn18} solved the problem in $O(n^{11}\log n)$ time.

\subsection{Our approach}

We follow the algorithm scheme in \cite{ref:OhVo20}, which in turn
follows that in~\cite{ref:AronovTh93}. Specifically, we first
compute the geodesic convex hull of all sites of $S$ in $O(n+m\log m)$
time~\cite{ref:GuibasOp89,ref:HershbergerA91,ref:ToussaintCo89}, and then
compute the geodesic center $c^*$ of the hull in $O(n+m)$ time~\cite{ref:AhnA16}.
Aronov~\cite{ref:AronovTh93} showed that the farthest Voronoi diagram
forms a tree with $c^*$ as the root and all leaves on $\partial P$, the boundary of
$P$. We construct the farthest Voronoi diagram restricted to
$\partial P$; this can
be done in $O(n+m)$ time by a recent algorithm of Oh et
al.~\cite{ref:OhTh20} once the geodesic convex hull of $S$ is known.

Next we run a {\em reverse geodesic sweeping} algorithm to extend the diagram from
$\partial P$ to the interior of $P$ (i.e., based on all
leaves on $\partial P$ and the root $c^*$ of the tree, we want to construct the tree). Here we use a {\em geodesic sweeping circle} that consists of all points with the
same geodesic distance from $c^*$. Aronov~\cite{ref:AronovTh93}
implemented this sweeping algorithm in $O((n+m)\log (n+m))$ time. Oh
and Ahn~\cite{ref:OhVo20} gave an improved solution of $O(n+m\log
m+m\log^2 n)$ time by using a data structure for the following query
problem: Given three points in $P$, compute the point that is equidistant
from them. Oh and Ahn~\cite{ref:OhVo20} built a data structure
in $O(n)$ time that can answer each query in $O(\log^2 n)$ time, and
that is why the time complexity of their algorithm has a $\log^2 n$ factor. We
improve the query time to $O(\log n)$ (with $O(n)$ time preprocessing) with the help of the following observations.
First, the three points involved in a query are
three sites of $S$ whose Voronoi cells are adjacent along the sweeping circle.
Second, among the three sites involved in a query, for every two sites whose Voronoi cells are adjacent, the
sweeping algorithm provides us with a point equidistant to them. These observations along with the tentative
prune-and-search technique of Kirkpatrick and
Snoeyink~\cite{ref:KirkpatrickTe95} lead us to a query algorithm
of $O(\log n)$ time. Consequently, the sweeping algorithm can be
implemented in $O(n+m\log m)$ time.

We should point out that in her
algorithm for computing the geodesic nearest-point Voronoi diagram,
Oh~\cite{ref:OhOp19} also announced an $O(\log n)$ time algorithm for
the above query problem and her algorithm also uses the tentative
prune-and-search technique (although the details are omitted due to the page limit).
However, the difference is that she uses a balanced geodesic
triangulation~\cite{ref:ChazelleRa94} and her result is based
on the assumption that the sought point of the query lies in a known
geodesic triangle $\triangle$ and the three query points are in the same subpolygon of $P$ separated by a side of $\triangle$ (see Lemma~4.2 in~\cite{ref:OhOp19}). For our
problem, we do not need the balanced geodesic triangulation and
do not have such an assumption. Instead, our algorithm relies on the
observations mentioned above.


The rest of the paper is organized as follows. Section~\ref{sec:pre} defines notation and introduces some concepts. The algorithm for constructing the geodesic Voronoi diagram is described in Section~\ref{sec:algo}. Section~\ref{sec:lemma} presents the algorithm for a lemma about the query problem discussed above.

\section{Preliminaries}
\label{sec:pre}

Like the previous work~\cite{ref:AronovTh93,ref:BarbaOp19,ref:OhTh20,ref:OhVo20}, for ease of discussion, we make a general position assumption that no vertex of $P$ is equidistant from two sites of $S$ and no point of $P$ has four farthest sites. We occasionally use {\em polygon vertex} to refer to a vertex of $P$ and use {\em polygon edge} to refer to an edge of $P$.

For any two points $p$ and $q$ in $P$, let $\pi(p,q)$ denote the (Euclidean) shortest path from $p$ to $q$ in $P$; let $d(p,q)$ denote the length of $\pi(p,q)$. $\pi(p,q)$ is also called the {\em geodesic path} and $d(p,q)$ is called the {\em geodesic distance} between $p$ and $q$. The vertex of $\pi(p,q)$ adjacent to $q$ (resp., $p$) is called the {\em anchor} of $q$ (resp., $p$) in $\pi(p,q)$.

For any two points $a$ and $b$ in the plane, denote by $\overline{ab}$ the line segment with $a$ and $b$ as endpoints, and denote by $|\overline{ab}|$ the length of the segment.


For any two sites $s$ and $t$ of $S$, their {\em bisector}, denoted by $B(s,t)$, consists of all points of $P$ equidistant from them, i.e., $B(s,t)=\{p\ |\ d(s,p)=d(t,p), p\in P\}$. Due to the general position assumption, Aronov et al.~\cite{ref:AronovOn89} showed that $B(s,t)$ is a smooth curve connecting two points on $\partial P$ with no other points common with $\partial P$ and $B(s,t)$ comprises $O(n)$ straight and hyperbolic arcs (a straight arc is a line segment); the endpoints of the arcs are {\em breakpoints}, each of which is the intersection of $B(s,t)$ and a segment extended from a polygon vertex $u$ to another polygon vertex $v$ such that $u$ is an anchor of $v$ in $\pi(s,v)$ or in $\pi(t,v)$ (it is possible that $u=s$ or $u=t$); e.g., see Fig.~\ref{fig:bisector}.

For any site $s\in S$, define $C(s)$ as the region consisting of all points $p$ of $P$ whose farthest site is $s$, i.e., $C(s)=\{p\ |\ d(p,s)\geq d(p,s'), s'\in S\}$. We call $C(s)$ the (farthest) {\em Voronoi cell} of $s$. Note that $C(s)$ may be empty; if $C(s)$ is not empty, then it is simply connected~\cite{ref:AronovTh93}.
The Voronoi cells of all sites of $S$ form a partition of $P$. We define the {\em geodesic farthest-point Voronoi diagram} (or {\em farthest Voronoi diagram} for short), denoted by $\fvd(S)$, as the closure of the interior of $P$ minus the union of the interior of $C(s)$ for all $s\in S$; alternatively, $\fvd(S)=\{p\in B(s,t) \ | \ s, t \in S \text{ and } d(s,p)=\max_{r\in S}d(r,p)\}$.
A point $v$ of $\fvd(S)$ is a {\em Voronoi vertex} if it is an intersection of a bisector with $\partial P$ or if it has degree $3$ (i.e., it has three equidistant sites). The curve of $\fvd(S)$ connecting two adjacent vertices is called a {\em Voronoi edge}, which is a portion of a bisector of two sites. Note that a Voronoi edge may not be of constant size because it may contain multiple breakpoints. While $\fvd(S)$ has $O(m)$ Voronoi vertices and edges, the total complexity of $\fvd(S)$ is $O(n+m)$~\cite{ref:AronovTh93}.

A subset $P'$ of $P$ is {\em geodesically convex} if $\pi(p,q)$ is in $P'$ for any two points $p$ and $q$ in $P'$. The {\em geodesic convex hull} of $S$ in $P$, denoted by $\gch(S)$, is the common intersection of all geodesically convex sets containing $S$. $\gch(S)$ is a weakly simple polygon of at most $n+m$ vertices.
Let $c^*$ be the geodesic center of $\gch(S)$, which
is also the geodesic center of $S$~\cite{ref:AronovTh93}. Note that $c^*$ must be on a Voronoi edge of $\fvd(S)$. Indeed, if $c^*$ has three farthest sites in $S$, then $c^*$ is
a Voronoi vertex; otherwise it has two farthest sites and thus is
in the interior of an edge of $\fvd(S)$. Aronov~\cite{ref:AronovTh93}
proved that $\fvd(S)$ is a tree with $c^*$ as the root and all leaves on $\partial P$; 
he also showed that only sites on the boundary of $\gch(s)$ have nonempty cells in $\fvd(S)$ and the ordering of the sites with nonempty cells around the boundary of $\gch(s)$ is the same as the ordering of their Voronoi cells around $\partial P$ (the {\em ordering lemma}).
Note that a site on the boundary of $\gch(s)$ may still have an empty Voronoi cell and intuitively this is because $P$ is not large enough~\cite{ref:AronovTh93}.

Consider any three points $s,t,r$ in $P$. The vertex farthest to $s$ in $\pi(s,t)\cap \pi(s,r)$ is called the {\em junction vertex} of $\pi(s,t)$ and $\pi(s,r)$. The closure of the interior of the geodesic convex hull $\gch(s,t,r)$ is called the {\em geodesic triangle} of $s$, $t$, and $r$, denoted by $\triangle(s,t,r)$, whose boundary is composed of three convex chains $\pi(s',t')$, $\pi(t',r')$, $\pi(r',s')$, where $s'$ is the junction vertex of $\pi(s,t)$ and $\pi(s,r)$, and $t'$ and $r'$ are defined likewise; e.g., see Fig.~\ref{fig:triangle}. The three convex chains are called {\em sides} of $\triangle(s,t,r)$. The three vertices $s'$, $t'$, and $r'$ are called the {\em apexes} of $\triangle(s,t,r)$.

\begin{figure}[t]
\begin{minipage}[t]{0.49\textwidth}
\begin{center}
\includegraphics[height=1.8in]{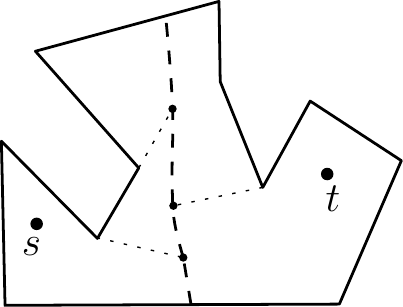}
\caption{\footnotesize Illustrating the bisector $B(s,t)$ (the dashed curve) with three breakpoints.}
\label{fig:bisector}
\end{center}
\end{minipage}
\hspace{0.05in}
\begin{minipage}[t]{0.49\textwidth}
\begin{center}
\includegraphics[height=1.8in]{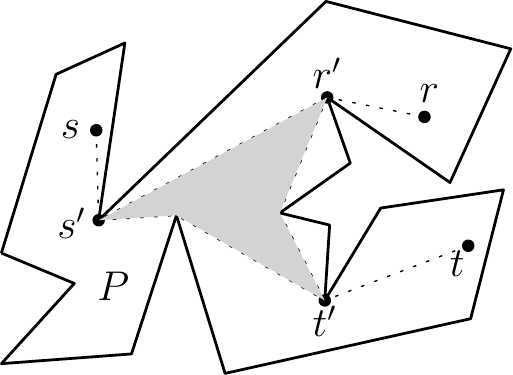}
\caption{\footnotesize Illustrating a geodesic triangle $\triangle(s,t,r)$ (the gray region).}
\label{fig:triangle}
\end{center}
\end{minipage}
\vspace{-0.15in}
\end{figure}


\section{Computing the farthest Voronoi diagram {\boldmath $\fvd(S)$}}
\label{sec:algo}

In this section, we present our algorithm for computing the farthest Voronoi diagram $\fvd(S)$.

First, we compute the geodesic convex hull $\gch(S)$ of $S$ in
$O(n+m\log m)$ time~\cite{ref:GuibasOp89,ref:HershbergerA91,ref:ToussaintCo89}.
Second, we compute the geodesic center $c^*$ of $\gch(S)$ in $O(n+m)$
time~\cite{ref:AhnA16}. Third, we compute the portion of $\fvd(S)$
restricted to the polygon boundary $\partial P$, i.e., the leaves of
$\fvd(S)$. This can be done in $O(n+m)$ time by the algorithm in~\cite{ref:OhTh20}.\footnote{Note that the result was not explicitly given in~\cite{ref:OhTh20} but can be obtained from their $O(n\log\log n+m\log m)$-time algorithm for computing $\fvd(S)$. Indeed, given $\gch(S)$, the algorithm first partitions $P$ in $O(n+m)$ time into $O(1)$ subpolygons such that each subpolygon $P'$ is for a problem instance where all involved sites are on the boundary of $P'$ (see Section 7~\cite{ref:OhTh20}). Then, each problem instance is further reduced in linear time to a problem instance where all sites are vertices of $P'$ (see Section 6~\cite{ref:OhTh20}), and each such problem instance can be solved in linear time (see Section 3~\cite{ref:OhTh20}).
The total running time of all above is $O(n+m)$ (for computing $\fvd(S)$ restricted to the boundary of $P$ only). This result was also used by Oh and Ahn~\cite{ref:OhVo20} in their $O(n+m\log m+m\log^2 n)$-time algorithm for computing $\fvd(S)$.}
The fourth step is to extend the diagram to the interior of $P$, i.e., construct the tree $\fvd(S)$ based on all
its leaves and the root $c^*$. This is achieved by a reverse
geodesic sweeping algorithm, whose details are described below.

The algorithm first computes the adjacency information of
$\fvd(S)$. Specifically, we will compute the locations of all Voronoi
vertices of $\fvd(S)$; for Voronoi edges, however, we will not compute
them exactly (i.e., the locations of their breakpoints will
not be computed) but only output their incident Voronoi vertices,
i.e., if $u$ and $v$ are two Voronoi vertices incident to the same
Voronoi edge, then we will output the pair $(u,v)$ as an {\em abstract} Voronoi edge. In this way, we
will output the abstract tree $\fvd(S)$ with the exact locations of
all Voronoi vertices; this is called the {\em topological structure} of
$\fvd(S)$ in~\cite{ref:OhVo20}. After having the topological
structure, Oh and Ahn~\cite{ref:OhVo20} gave an algorithm that can
construct $\fvd(S)$ in additional $O(n+m\log m)$ time. More
specifically, with $O(n)$ time preprocessing, each Voronoi edge can be
computed in $O(\log n + k)$ time, where $k$ is the number of
breakpoints in the Voronoi edge (see Section~4 of~\cite{ref:OhVo20}
for details). As $\fvd(S)$ has $O(m)$ Voronoi edges, the total time
for computing all Voronoi edges is $O(m\log n + K)$, where $K$ is the
	total number of breakpoints on all Voronoi edges. As $K=O(n+m)$~\cite{ref:AronovTh93}, the
		total time is bounded by $O(n+m\log n)$, which is $O(n+m\log
		m)$.\footnote{Indeed, if $m<n/\log n$, then $n+m\log
		n=\Theta(n)$, which is $O(n+m\log m)$; otherwise, $\log
		n=O(\log m)$ and $n+m\log n=O(n+m\log m)$.}
In the following, we will focus on computing the topological structure
of $\fvd(S)$.

We use a reverse geodesic sweeping as in~\cite{ref:AronovTh93,ref:OhVo20}. Roughly speaking, the sweep line is a
{\em geodesic circle} $C$ consisting of all points in $P$ that have the same geodesic
distance from the geodesic center $c^*$ of $S$. This statement is actually not quite accurate as
initially the sweep circle is just the boundary of $P$. During the
sweeping, we maintain the sites whose Voronoi cells currently
intersect $C$; these sites are stored in a cyclic linked list $\calL$ ordered
by their intersections with $C$. Initially when $C=\partial P$, as we already have the leaves of
$\fvd(S)$, we can build $\calL$ in $O(n+m)$ time. Note that
$|\calL|=O(m)$. During the algorithm, $C$ will
shrink until $c^*$; an event happens when $C$ hits a Voronoi vertex, which will be computed on the fly.
Specifically, for each triple of adjacent sites $s,t,r$ in the list
$\calL$, we compute the point, denoted by $\alpha(s,t,r)$, equidistant from them, which is
the intersection of the bisectors $B(s,t)$ and $B(t,r)$. Due to our general position assumption, $\alpha(s,t,r)$ is unique if it exists (see Lemma~2.5.3~\cite{ref:AronovTh93}).
We store all
these $\alpha$-points in a priority queue $Q$, ordered by decreasing geodesic
distance from $c^*$. In order to compute the $\alpha$-points, for any pair
of adjacent sites $s$ and $t$ in $\calL$, we maintain a Voronoi vertex, denoted by $\beta(s,t)$, on their bisector $B(s,t)$ with the following property: $\beta(s,t)$ is outside or on the current geodesic circle $C$. Initially, we set $\beta(s,t)$ to be the Voronoi vertex on $\partial P$ incident to the Voronoi cells of $s$ and $t$; so the above property holds as $C=\partial P$.

The main loop of the algorithm works as follows.
As long as $Q$ is not empty, we repeatedly extract the point with
largest geodesic distance from $c^*$ and let
the point be $\alpha(s,t,r)$ defined by three sites $s,t,r$ in this
order in $\calL$. We report $\alpha(s,t,r)$ as a Voronoi vertex and report
$(\beta(s,t),\alpha(s,t,r))$ and $(\beta(t,r),\alpha(s,t,r))$ as two
abstract Voronoi edges.
We remove $t$ from $\calL$ and set $\beta(s,r)=\alpha(s,t,r)$. Let $x$ be the neighbor of
$s$ other than $r$ in $\calL$ and let $y$ be the neighbor of $r$
other than $s$. We remove $\alpha(x,s,t)$ and $\alpha(t,r,y)$
from $Q$ if they exist. Next, we compute $\alpha(x,s,r)$ and $\alpha(s,r,y)$ (if exist) as well as their geodesic distances from $c^*$, and insert them into $Q$.

For the running time, there are $O(m)$ events, for the total number of Voronoi vertices of $\fvd(S)$ is $O(m)$~\cite{ref:AronovTh93}, and thus the total time of the algorithm is
$O(m\cdot \sigma)$, where $O(\sigma)$ is the time for computing each $\alpha$ point. Lemma~\ref{lem:10}, which will be proved later in Section~\ref{sec:lemma}, is for computing the $\alpha$-points.

\begin{lemma}\label{lem:10}
With $O(n)$ time preprocessing, for any triple of adjacent sites $s,t,r$ in $\calL$ at any moment during the algorithm, given the two Voronoi vertices $\beta(s,t)$ and $\beta(t,r)$, our algorithm can do the following in $O(\log n)$ time: if $\alpha(s,t,r)$ is a Voronoi vertex, then compute it; otherwise, either compute $\alpha(s,t,r)$ or return null.
\end{lemma}

We remark that Lemma~\ref{lem:10} is sufficient for the correctness of our geodesic sweeping algorithm as only Voronoi vertices are essential. If the algorithm returns null, the event will not be inserted to $Q$.
With Lemma~\ref{lem:10} at hand, our geodesic sweeping algorithm computes the topological structure of $\fvd(S)$ in $O(n+m\log m)$ time. After that, as discussed above, we can compute the full diagram $\fvd(S)$ in additional $O(n+m\log m)$ time by the techniques of Oh and Ahn~\cite{ref:OhVo20}. Also, the space of the algorithm is bounded by $O(n+m)$.

\begin{theorem}\label{theo:10}
The geodesic farthest-point Voronoi diagram of a set of $m$ points in a simple polygon of $n$ vertices can be computed in $O(n+ m\log m)$ time and $O(n+m)$ space.
\end{theorem}

\section{Algorithm for Lemma~\ref{lem:10}}
\label{sec:lemma}

In this section, we present our algorithm for Lemma~\ref{lem:10}.
We first present an algorithm in Section~\ref{subsec:triquery} for the following {\em triple-point geodesic center query} problem: given any three points in $P$, compute their geodesic center
in $P$, which is a point that minimizes the largest geodesic distance
from the three query points. Oh and Ahn~\cite{ref:OhVo20} solved this problem in $O(\log^2 n)$ time,
after $O(n)$ time preprocessing. Our algorithm runs in $O(\log n)$ time also with $O(n)$ time preprocessing.\footnote{To be fair, this problem is not a dominant one in their algorithm, which might be a reason Oh and Ahn~\cite{ref:OhVo20} did not push their result further.}
This algorithm will be used as a subroutine in our algorithm for Lemma~\ref{lem:10}, which will be discussed in Section~\ref{subsec:lemma}.


\subsection{The triple-point geodesic center query problem}
\label{subsec:triquery}

As preprocessing, we construct the two-point shortest path query data
structure by Guibas and Hershberger~\cite{ref:GuibasOp89,ref:HershbergerA91} and we refer to
it as the {\em GH data structure}. The data structure can be constructed in
$O(n)$ time, after which given any two points $p$ and $q$ in $P$, the
geodesic distance $d(p,q)$ can be computed in $O(\log n)$ time and the
geodesic path $\pi(p,q)$ can be output in additional time linear in
the number of edges of $\pi(p,q)$.

Consider three query points $s$, $t$, and $r$ in $P$. Our goal is to compute
their geodesic center, denoted by $c$. We follow the algorithm scheme in~\cite{ref:OhVo20}. Consider the geodesic convex hull $\gch(s,t,r)$ and the geodesic triangle $\triangle(s,t,r)$. We know that $c$ is the
geodesic center of $\gch(s,t,r)$~\cite{ref:AronovTh93}. Depending on
whether $c$ is in the interior of $\triangle(s,t,r)$, there are two cases.

\subsubsection{$c$ is not in the interior of $\triangle(s,t,r)$}
If $c$ is not in the interior of $\triangle(s,t,r)$, then it must be on the geodesic path of
two points of $\{s, t, r\}$. Without loss of generality, we assume
that $c\in \pi(s,t)$. Note that $c$ must be the middle point of
$\pi(s,t)$. To locate $c$ in $\pi(s,t)$, we wish to do
binary search on the vertices of $\pi(s,t)$. It was claimed in~\cite{ref:OhVo20} that the query algorithm
of the GH data structure returns $\pi(s,t)$ as a binary tree (so that binary search can be done in a straightforward way), in particular, when the simpler approach
in~\cite{ref:HershbergerA91} is utilized. In fact, this is not quite
correct. Indeed, the binary tree structures in both
\cite{ref:GuibasOp89} and \cite{ref:HershbergerA91} are used for
representing convex chains (or more rigorously, {\em semiconvex chains}~\cite{ref:HershbergerA91}).
However, $\pi(s,t)$ is actually a {\em string}~\cite{ref:GuibasOp89}, which in general is not a
semiconvex chain. The data structure for representing a string is a
tree but not a binary tree because a node in the tree may have three
children.

Here for completeness, we provide a general binary search scheme on the geodesic path $\pi(p,q)$ returned by the GH data structure for any two query points $p$ and $q$ in $P$.
Suppose we are looking for either a vertex or an edge of $\pi(p,q)$, denoted by $w^*$ in either case, and we have access to an {\em oracle} such that given
any vertex $v\in \pi(p,q)$, the oracle can determine whether $w^*$ is in
$\pi(p,v)$ or in $\pi(v,q)$. Then, we have the following lemma.

\begin{lemma}\label{lem:20}
With $O(n)$ time preprocessing, given any two query points $p$ and
$q$, the sought vertex or edge $w^*$ can be located by a binary search algorithm that
calls the oracle on $O(\log n)$ vertices of $\pi(p,q)$, and the total
time of the binary search excluding the time for calling the oracle is
$O(\log n)$. In particular, the middle point of $\pi(p,q)$ can be
found in $O(\log n)$ time.
\end{lemma}
\begin{proof}
We will use notation and concepts from the GH data structure~\cite{ref:GuibasOp89} without much explanation. To represent convex chains, we utilize the simper way given in~\cite{ref:HershbergerA91}, i.e., persistent binary trees with the path-copying method.

Consider the query points $p$ and $q$. The GH query algorithm combines $O(\log\log n)$ ``small'' {\em hourglasses} of size $O(\log n)$ and two ``big'' hourglasses of size $O(n)$ to assemble the path $\pi(p,q)$, which is represented as a {\em string}~\cite{ref:GuibasOp89}. Combining two hourglasses involves computing a tangent between them such that the tangent belongs to $\pi(p,q)$.
Thus, the algorithm will produce $O(\log\log n)$ tangents.
For our problem, we explicitly consider these tangents and call the oracle on every vertex of these tangents. This calls the oracle $O(\log\log n)$ times. After that we can determine an hourglass that contains $w^*$. If it is a small hourglass, then since it has $O(\log n)$ vertices, we can simply call the oracle on every vertex to locate $w^*$. In the following, we assume that $w^*$ is in a big hourglass and let $\pi$ be the portion of $\pi(p,q)$ in the hourglass.

The endpoints of $\pi$ can be obtained during the GH query algorithm. $\pi$ from one end to the other consists of a convex chain, a string, and another convex chain in order. The two connection vertices between the string and the two convex chains can be maintained during the preprocessing. We call the oracle on these two vertices, after which we can determine the one of the three portions of $\pi$ that contains $w^*$. If it is a convex chain, then as a convex chain is represented by a binary tree of height $O(\log n)$~\cite{ref:HershbergerA91}, we can apply binary search on this tree in a standard way; after calling the oracle on $O(\log n)$ vertices, $w^*$ can be obtained. In the following, we assume that the string of $\pi$ contains $w^*$. By slightly abusing notation, we still use $\pi$ to denote the string.

The string $\pi$ is represented by a tree $T$. However, each node of $T$ may have three children. Consider the root $v$ of $T$. In general, $\pi$ is a {\em derived string} and $v$ has three children: a left subtree $L$ representing a derived string, a middle subtree $M$ representing a {\em fundamental string}, and a right subtree $R$ representing another derived string. The fundamental string consists of two convex chains linked by a tangent edge and each derived string consists of two derived strings and a fundamental string in the middle. The height of $T$ is $O(\log n)$. The connecting vertex between the left string and the middle string and the connecting vertex between the middle string and the right string are maintained in the preprocessing and thus available during the query algorithm. We call the oracle on the two connecting vertices and then determine which string contains $w^*$. If the left or the right string contains $w^*$, then we proceed on the corresponding subtree of $v$ recursively. Otherwise, the middle string contains $w^*$. Again, the middle string consists of two convex chains linked by a tangent edge, which is available to us due to the preprocessing. We call the oracle on the two vertices of the tangent edge to determine which convex chain contains $w^*$ (or whether the tangent edge contains $w^*$). After that, since a convex chain is represented by a binary tree of height $O(\log n)$, we can finally locate $w^*$ by calling the oracle on $O(\log n)$ vertices using the binary tree. In this way, after $O(\log n)$ oracle calls, we can reach a fundament string and then $w^*$ can be finally located after another $O(\log n)$ oracle calls.

In summary, by calling the oracle on $O(\log n)$ vertices of $\pi(p,q)$, $w^*$ can be found.

To compute the middle point of $\pi(p,q)$, we first compute the geodesic distance $d(p,q)$ in $O(\log n)$ time by using the GH data structure. Then, we can follow the above binary search scheme and each time when the oracle is called on a vertex $v$, we also keep track of the geodesic distance $d(p,v)$ using the GH data structure. By comparing $d(p,v)$ with $d(p,q)/2$, we can decide which way to proceed the search. In this way, the middle point of $\pi(p,q)$ can be determined in $O(\log n)$ time.
\end{proof}

With Lemma~\ref{lem:20} at hand, we can find $c$ on $\pi(s,t)$ in $O(\log n)$ time.

We can determine whether $c$ is in the interior of $\triangle(s,t,r)$
in $O(\log n)$ time using Lemma~\ref{lem:20} as follows. First, we
determine whether $c$ is the middle point of $\pi(s,t)$. To do so,
we first compute the middle point $p_{st}$ of $\pi(s,t)$ by
Lemma~\ref{lem:20}. Then, we compute $d(s,p_{st})$ and $d(r,p_{st})$ in $O(\log n)$ time using the GH data structure. It is not difficult to see that $p_{st}$ is $c$ if and only if  $d(s,p_{st})\geq d(r,p_{st})$. If
$p_{st}\neq c$, then we use the same way to determine whether the middle point
of $\pi(s,r)$ (resp., $\pi(r,t)$) is $c$. If the above algorithm
fails to locate $c$, then we know that $c$ is in the interior of $\triangle(s,t,r)$.

The above finds $c$ in $O(\log n)$ time for the case where $c$ is not in the interior of $\triangle(s,t,r)$.

\subsubsection{$c$ is in the interior of $\triangle(s,t,r)$}
We proceed to the case where $c$ is in the interior of
$\triangle(s,t,r)$.
Our algorithm utilizes the tentative prune-and-search technique of Kirkpatrick
and Snoeyink~\cite{ref:KirkpatrickTe95}.

First observe that in this case $c$ must be equidistant from all
three points $s$, $t$, and $r$. Let $s'$ be the junction vertex of $\pi(s,t)$ and $\pi(s,r)$. Define $t'$ and $r'$ similarly. With the GH data structure, each junction vertex can be computed in $O(\log n)$ time~\cite{ref:GuibasOp89}.

Define $p_s$, $p_t$, and $p_r$ to be the anchors of $c$ in $\pi(s,c)$, $\pi(t,c)$, and $\pi(r,c)$, respectively. Note that the segment connecting $c$ to $p_s$ (resp., $p_t$, $p_r$) is tangent to the side of $\triangle(s,t,r)$ that contains it.
As $c$ is equidistant to $s$, $t$, and $r$, $c$ is the common intersection of the three bisectors $B(s,t)$, $B(s,r)$, and $B(t,r)$.

\begin{observation}\label{obser:00}
The middle point $p_{st}$ of $\pi(s,t)$ must be in $\pi(s',t')$; the middle point $p_{sr}$ of
$\pi(s,r)$ must be in $\pi(s',r')$; the middle point $p_{tr}$ of $\pi(t,r)$ must be in $\pi(t',r')$.
\end{observation}
\begin{proof}
We only prove the case for $p_{st}$ since the other two cases are similar. Assume to the contrary that $p_{st}\not\in \pi(s',t')$. Then, either $p_{st}\in \pi(s,s')\setminus\{s'\}$ or $p_{st}\in \pi(t,t')\setminus\{t'\}$. We assume it is the former case as the analysis for the latter case is similar. Then, $d(s,s')>d(t,s')$. Note that $d(s,c)=d(s,s')+d(s',c)$ as $c$ is in the interior of $\triangle(s,t,r)$. Hence, $d(s,c)>d(t,s')+d(s',c)\geq d(t,c)$. But this incurs contradiction as $c$ is equidistant from $s$ and $t$.
\end{proof}

Since $p_s$ is the anchor of $c$ in $\pi(s,c)$, $p_s$ has smaller geodesic distance from $s$ than from $t$ or $r$. Hence, by Observation~\ref{obser:00}, $p_s$ must be in $\gamma_s=\pi(s',p_{st})\cup \pi(s',p_{sr})$, which consists of two convex chains; we call $\gamma_s$ a {\em pseudo-convex chain}.
Similarly, $p_t$ must be in $\gamma_t=\pi(t',p_{st})\cup \pi(t',p_{tr})$ and $p_r$ must be in $\gamma_r=\pi(r',p_{tr})\cup \pi(r',p_{sr})$. We consider $p_{st}$ and $p_{sr}$ as two ends of $\gamma_s$. If we move a point $p$ on $\gamma_s$ from one end to the other, then the slope of the tangent line of $\gamma_s$ at $p$ continuously changes. Further, $p_s$ is the only point on $\gamma_s$ such that
$\overline{cp_s}$ is tangent to $\gamma_s$. Similar properties
hold for $\gamma_t$, $p_t$, $\gamma_r$, and $p_r$.

With the above discussion, we are now in a position to describe our
algorithm for computing $c$. First, we compute the three junction
vertices and the three middle points $s'$,
$t'$, $r'$, $p_{st}$, $p_{sr}$, and $p_{tr}$. This can be done in
$O(\log n)$ time using the GH data structure.
To compute $c$ (as well as
locate $p_s$, $p_t$, and $p_r$), we resort to the tentative prune-and-search technique~\cite{ref:KirkpatrickTe95}, as follows.

To avoid the lengthy background explanation, we follow the notation
in~\cite{ref:KirkpatrickTe95} without definition.
We will rely on Theorem 3.9 in~\cite{ref:KirkpatrickTe95}. To this
end, we need to define three continuous and monotone-decreasing functions $f$, $g$,
and $h$. We define them in a way similar in spirit to Theorem 4.10
in~\cite{ref:KirkpatrickTe95} for finding a point equidistant to three
convex polygons. Indeed, our problem may be considered as a weighted
case of their problem because each point in our pseudo-convex chains has a
weight that is equal to its geodesic distance from one of $s$, $t$,
and $r$.

\begin{figure}[t]
\begin{minipage}[t]{\textwidth}
\begin{center}
\includegraphics[height=1.8in]{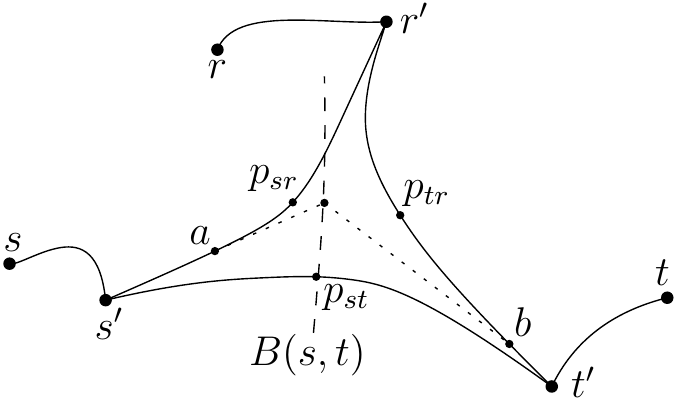}
\caption{\footnotesize Illustrating the geodesic triangle $\triangle(s,t,r)$ and the definition of the function $f(a)$ for $a\in A=\gamma_s$.}
\label{fig:function}
\end{center}
\end{minipage}
\vspace{-0.15in}
\end{figure}

We parameterize over $[0,1]$ each of the three pseudo-convex chains $A=\gamma_s$,
$B=\gamma_t$, and $C=\gamma_r$ from one end to the other in
counterclockwise order around $\triangle(s,t,r)$. For example, without loss
of generality, we assume that $s'$, $t'$, and $r'$ are
counterclockwise around $\triangle(s,t,r)$. Then, $\gamma_s$ is
parameterized from $p_{sr}$ to $p_{st}$ over $[0,1]$, i.e., each
value of $[0,1]$ corresponds to a slope of a tangent at a point on
$\gamma_s$. For each point $a$ of $A$, we define
$f(a)$ to be the parameter of the point $b\in B$ such that the tangent
of $A$ at $a$ and the tangent of $B$ at $b$ intersect at a point on the bisector $B(s,t)$ of $s$ and $t$ (e.g., see Fig.~\ref{fig:function}). Similarly, we define
$g(b)$ for $b\in B$ with respect to $C$ and define $h(c)$ for $c\in C$
with respect to $A$. One can verify that all three functions are
continuous and monotone-decreasing (the tangent at an apex of $\triangle(s,t,r)$ is not unique but the issue can be handled~\cite{ref:KirkpatrickTe95}). The fixed-point of the
composition of the three functions $h\cdot g \cdot f$ corresponds to
$c$, which can be computed by applying the tentative prune-and-search
algorithm of Theorem 3.9~\cite{ref:KirkpatrickTe95}.

To see that the algorithm can be implemented in $O(\log n)$ time,  we need to show that given
any $a\in A$ and any $b\in B$, we can determine whether $f(a)> b$ in $O(1)$
time. To this end, we first find the intersection $p$ of the tangent of $A$ at
$a$ and the tangent of $B$ at $b$. Then, $d(s,a)+|\overline{pa}|<d(t,b)+|\overline{pb}|$ if and only if
$f(a)>b$. We will discuss below that the values $d(s,a)$ and $d(t,b)$ will be available during the tentative prune-and-search algorithm. Note that here the tangent of $A$ at $a$ actually refers to the half-line of the tangent whose concatenation with $\pi(s',a)$ is still a convex chain (so that the shortest path can follow that half-line), as shown in Fig.~\ref{fig:function}. Hence, it is possible that the tangent half-line of
$a$ does not intersect the tangent half-line of $b$. If that happens, either the tangent half-line of
$a$ intersects the backward extension of the tangent half-line of $b$ or the backward extension of the tangent half-line of $a$ intersects the tangent half-line of $b$; in the former case we have $f(a)<b$ and in the latter case $f(a)>b$. Similar properties hold for functions of $g$ and $h$.
Finally, we show that we have appropriate data structures to represent
the three pseudo-convex chains $A$, $B$, and $C$ so that the algorithm can terminate
in $O(\log n)$ rounds. We only discuss $A$ since other two
are similar. When the algorithm picks the first vertex of $A$ to test,
we will use the vertex $s'$. After the test, the algorithm will
proceed on $A$ on one side of $s'$, say, on $\pi(s',p_{st})$.
We apply the binary search scheme of Lemma~\ref{lem:20} on
$\pi(s',p_{st})$, which will test $O(\log n)$ vertices. Further,
whenever a vertex $a\in \pi(s',p_{st})$ is tested, the binary search
scheme of Lemma~\ref{lem:20} can keep track of $d(s,a)$.
Therefore, applying the tentative prune-and-search
technique in Theorem 3.9~\cite{ref:KirkpatrickTe95} can compute the
geodesic center $c$ in $O(\log n)$ time.

The following lemma summarizes our result on the triple-point geodesic
center query problem.

\begin{lemma}\label{lem:centerquery}
With $O(n)$ time preprocessing, the geodesic center of any
three query points in $P$ can be computed in $O(\log n)$ time.
\end{lemma}

\subsection{Proving Lemma~\ref{lem:10}}
\label{subsec:lemma}

With Lemma~\ref{lem:centerquery}, we are ready to present our algorithm for Lemma~\ref{lem:10}. Consider any three sites $s$, $t$, $r$ as specified in the statement of Lemma~\ref{lem:10}.
Our goal is to compute the point $\alpha(s,t,r)$, which is equidistant from
the three sites. Recall that we have two points $\beta(s,t)$ and $\beta(t,r)$ available for us, which are critical to the success of our approach.

The first step of our algorithm is to apply the algorithm for
Lemma~\ref{lem:centerquery} to compute the geodesic center $c$ of the three
sites. We check whether $c$ is equidistant to the three sites,
in $O(\log n)$ time. If yes,
$\alpha(s,t,r)=c$ and we are done. In the following we assume
otherwise.

Our algorithm may not compute $\alpha(s,t,r)$ even if it exists, but
will guarantee to do so if $\alpha(s,t,r)$ is a Voronoi vertex of $\fvd(S)$.
This is sufficient for constructing $\fvd(S)$ correctly. Hence, in
what follows we assume that $\alpha(s,t,r)$ is a Voronoi vertex.
This implies that there are two Voronoi edges connecting
$\alpha(s,t,r)$ with $\beta(s,t)$ and $\beta(t,r)$, respectively. Recall that $c^*$ is the geodesic center of $S$. The following observation was discovered by Aronov et al.~\cite{ref:AronovTh93}.

\begin{observation}\label{obser:monotone}{\em (Aronov et al.~\cite{ref:AronovTh93})}
If a point $p$ moves from $\beta(s,t)$ (resp., $\beta(t,r)$) to $\alpha(s,t,r)$ along the Voronoi edge, both $d(c^*,p)$ and $d(t,p)$
are monotonically decreasing.
\end{observation}

To simplify the notation, unless otherwise stated, we use $\alpha$ to refer to $\alpha(s,t,r)$. We define the three junction vertices $s'$, $t'$, and $r'$ in the same way as before, which are the three apexes of the geodesic convex hull $\triangle(s,t,r)$. Without loss of generality, we assume that $s'$, $t'$, and $r'$ are counterclockwise around the boundary of $\triangle(s,t,r)$ (e.g., see Fig.~\ref{fig:triangle}).
 We define $p_s$, $p_t$, and $p_r$ as the anchors of $\alpha$ in $\pi(s,\alpha)$, $\pi(t,\alpha)$, and $\pi(r,\alpha)$, respectively.
Define $p_{st}$, $p_{sr}$, and $p_{tr}$ as the middle points of $\pi(s,t)$, $\pi(s,r)$, and $\pi(t,r)$, respectively.

The following observation, obtained from the results of Aronov~\cite{ref:AronovOn89}, will occasionally be used later.

\begin{observation}\label{obser:old}{\em (Aronov~\cite{ref:AronovOn89})}
Suppose $x$ and $y$ are two points of $P$ such that their bisector $B(x,y)$ does not contain any vertex of $P$. Then, for any point $z\in P$, the shortest path $\pi(x,z)$ (resp., $\pi(y,z)$) either does not intersect $B(x,y)$ or intersects it at a single point.
\end{observation}
\begin{proof}
All arguments here are from Aronov~\cite{ref:AronovOn89}.
$B(x,y)$ divides $P$ into two subpolygons; one of them, denoted by $P_x$, contains $x$ and the other, denoted by $P_y$, contains $y$. We assume that neither $P_x$ nor $P_y$ contains $B(x,y)$.
All points in $P_x$ are closer to $x$ than to $y$ and all points in $P_y$ are closer to $y$ than to $x$.
Let $z$ be any point in $P$. If $z\in P_x$, then $\pi(x,z)$ is in $P_x$. If $z\in B(x,y)$, then $\pi(x,z)\setminus\{z\}$ is in $P_x$ due to the general position assumption. If $z\in P_y$, then $\pi(x,z)$ intersects $B(x,y)$ at a single point. Similar results hold for $\pi(y,z)$.
\end{proof}

Recall that $\alpha(s,t,r)$ is not $c$. Hence, $c$ must be equidistant to two sites and
the geodesic distance from them to $c$ is strictly larger than that from the third
site to $c$. Depending on what the two sites are, there are three cases
$d(c,s)=d(c,r)>d(c,t)$, $d(c,t)=d(c,r)>d(c,s)$, and
$d(c,t)=d(c,s)>d(c,r)$.
The following lemma shows that the latter two cases cannot happen.

\begin{lemma}\label{lem:twocases}
Neither $d(c,t)=d(c,r)>d(c,s)$ nor $d(c,t)=d(c,s)>d(c,r)$ can happen.
\end{lemma}
\begin{proof}
Note that the two cases are symmetric and thus we only discuss the case $d(c,t)=d(c,r)>d(c,s)$.

Assume to the contrary that $d(c,t)=d(c,r)>d(c,s)$ happens. Then, $c$ is the middle point of $\pi(t,r)$. Consider the farthest Voronoi diagram $\fvd(s,t,r)$ with respect to the three sites $s,t,r$ only (without considering other sites of $S$). Then, $\alpha$ is a vertex of the diagram, i.e., the three Voronoi edges bounding the three cells of $s$, $t$, and $r$ meet at $\alpha$ (e.g., see Fig.~\ref{fig:twocases}). Since $d(c,t)=d(c,r)>d(c,s)$, $c$ is on the Voronoi edge $E(t,r)$ bounding the cells of $t$ and $r$. By the definition of $\beta(t,r)$, it is also on $E(t,r)$. Hence, all three points $\alpha$, $c$, and $\beta(t,r)$ are on $E(t,r)$. Let $E'(t,r)$ be the portion of $E(t,r)$ between $\alpha$ and $\beta(t,r)$. Since both $\beta(t,r)$ and $\alpha$ are vertices of $\fvd(S)$, $E'(t,r)$ must be an edge of $\fvd(S)$ bounding the two cells of $t$ and $r$. By Observation~\ref{obser:monotone}, if we move a point $p$ on $E'(t,r)$ from $\beta(t,r)$ to $\alpha$, $d(t,p)$ is monotonically decreasing.

\begin{figure}[t]
\begin{minipage}[t]{\textwidth}
\begin{center}
\includegraphics[height=1.7in]{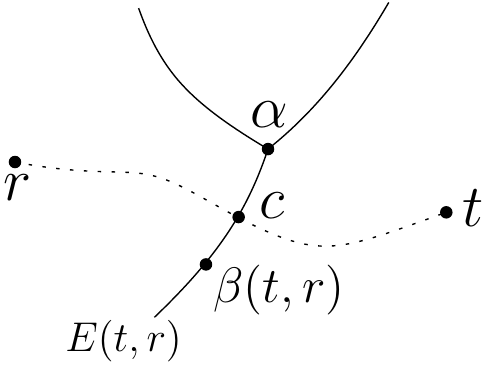}
\caption{\footnotesize Illustrating the proof of Lemma~\ref{lem:twocases}. The three solid curves are the three Voronoi edges of $\fvd(s,t,r)$. The dotted curve is the shortest path $\pi(r,t)$.}
\label{fig:twocases}
\end{center}
\end{minipage}
\vspace{-0.15in}
\end{figure}

Since $c$ is the middle point of $\pi(t,r)$, if we move a point $p$ on the bisector $B(t,r)$ from one end to the other, $d(t,p)$ will first striclty decreases until $p=c$ and then strictly increases. Note that $E'(t,r)\subseteq E(t,r) \subseteq B(t,r)$. Note also that $E'(t,r)$ has $\alpha$ and $\beta(t,r)$ as its two endpoints. Depending on whether $E'(t,r)$ contains $c$, there are two cases.

\begin{itemize}
\item
If $E'(t,r)$ does not contain $c$, then $c$, $\beta(t,r)$, and $\alpha$ appear in $E(t,r)$ in this order since $\alpha$ is a vertex of $\fvd(s,t,r)$ and thus is an endpoint of $E(t,r)$. Hence, $c$, $\beta(t,r)$, and $\alpha$ appear in $B(t,r)$ in this order. Therefore, if we move a point $p$ on $E'(t,r)$ from $\beta(t,r)$ to $\alpha$, $d(t,p)$ must be strictly increasing. But this contradicts with the fact that if we move a point $p$ on $E'(t,r)$ from $\beta(t,r)$ to $\alpha$, $d(t,p)$ is monotonically decreasing.

\item
If $E'(t,r)$ contains $c$, then $\beta(t,r)$, $c$, and $\alpha$ appear in $E(t,r)$ in this order (e.g., see Fig.~\ref{fig:twocases}). Hence, $\beta(t,r)$, $c$, and $\alpha$ appear in $B(t,r)$ in this order. Therefore, if we move a point $p$ on $E'(t,r)$ from $\beta(t,r)$ to $\alpha$, $d(t,p)$ will first strictly decrease and then strictly increase, a contradiction again.
\end{itemize}
The lemma thus follows.
\end{proof}

Note that Lemma~\ref{lem:twocases} is obtained based on the assumption that $\alpha(s,t,r)$ is a Voronoi vertex of $\fvd(S)$. Therefore, if one of the two cases in Lemma~\ref{lem:twocases} happens during the algorithm, then we can simply return null.


\medskip
In what follows, we assume that $d(c,s)=d(c,r)>d(c,t)$. Thus,
$c$ must be the middle point of $\pi(s,r)$. Depending on
where the location of $c$ is, there are three cases: $c\in
\pi(s',r')$, $c\in \pi(s,s')\setminus\{s'\}$, and $c\in
\pi(r',r)\setminus\{r'\}$. The latter
two cases are symmetric, so we will only discuss the first two
cases.

\subsubsection{The case $c\in \pi(s',r')$}
Let $\overline{uv}$ be the edge of $\pi(s',r')$ containing $c$ such that $d(s,u)<d(s,v)$. It possible that $u$ is $s$ or/and $v$ is $t$. We first assume that both $u$ and $v$ are polygon vertices; we will show later the other case can be reduced to this case.
Our algorithm relies on the following lemma (e.g., see Fig.~\ref{fig:subcase10lem}).

\begin{figure}[t]
\begin{minipage}[t]{\textwidth}
\begin{center}
\includegraphics[height=2.5in]{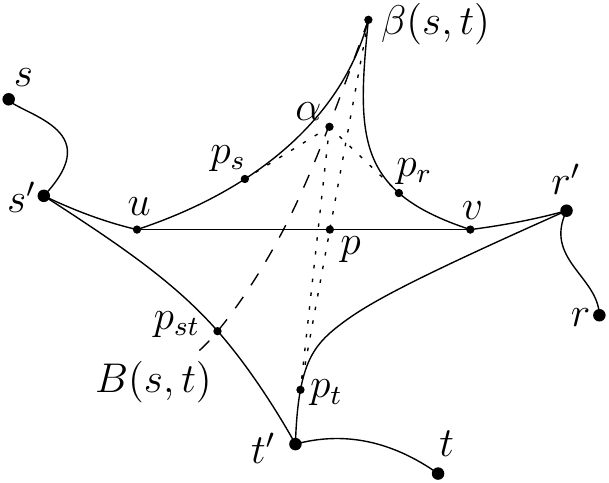}
\caption{\footnotesize Illustrating Lemma~\ref{lem:30}. In this example, $u'=u$ and $v'=v$.}
\label{fig:subcase10lem}
\end{center}
\end{minipage}
\vspace{-0.15in}
\end{figure}

\begin{lemma}\label{lem:30}
\begin{enumerate}
\item
$\alpha$ must be in the geodesic triangle $\triangle(s,r,\beta(s,t))$.

\item
The apexes of $\triangle(s,r,\beta(s,t))$ are $u'$, $v'$, and $\beta(s,t)$, where $u'$ (resp., $v'$) is the junction vertex of and $\pi(s,r)$ and $\pi(s,\beta(s,t))$ (resp., $\pi(r,\beta(s,t))$) (in Fig.~\ref{fig:subcase10lem}, $u'=u$ and $v'=v$).


\item
$p_s$ must be on the pseudo-convex chain $\pi(u',\beta(s,t))\cup \pi(u',v')$ and $\overline{\alpha p_s}$ is
tangent to the chain.

\item
$p_r$ must be on the pseudo-convex chain $\pi(v',\beta(s,t))\cup\pi(v',u')$ and $\overline{\alpha p_r}$ is tangent to the chain.

\item
$p_t$ must be on the pseudo-convex chain $\pi(t_{uv},u)\cup \pi(t_{uv},v)$ and $\overline{\alpha p_t}$ is
tangent to the chain, where $t_{uv}$ is the junction vertex of $\pi(t,u)$ and $\pi(t,v)$.

\item
$\overline{\alpha p_t}$ intersects $\overline{uv}$.
\end{enumerate}
\end{lemma}
\begin{proof}
As both $u$ and $v$ are polygon vertices, $\overline{uv}$ divides $P$ into two sub-polygons; one of them, denoted by $P_1$, does not contain $t$ and we use $P_2$ to denote the other one. Oh and Ahn~\cite{ref:OhVo20} claimed without proof (in the proof of Lemma 3.6~\cite{ref:OhVo20}) that $\alpha$ is in $P_1$. We provide a brief proof below.

Assume to the contrary that $\alpha$ is not in $P_1$. Then, $\alpha\in P_2$. Since $B(s,r)$ intersects $\pi(s,r)$ only once at $c\in \overline{uv}\subseteq \pi(s,r)$, $B(s,r)$ is partitioned into two portions by $c$, one in $P_1$ and the other in $P_2$; let $B_2(s,r)$ denote the portion in $P_2$, which has $c$ as one of its endpoint. As $\alpha$ is equidistant from $s$, $t$, and $r$, $\alpha$ is on $B(s,t)$. Since $\alpha\in P_2$, we obtain that $\alpha\in B_2(s,r)$. As $\alpha$ is not $c$, which is the geodesic center of $s$, $t$, and $r$, $\alpha$ cannot be in the geodesic triangle $\triangle(s,t,r)$. Therefore, if we move on $B_2(s,r)$ from $c$ to its other endpoint, we will first enter $\triangle(s,t,r)$ and then encounter either $\pi(s',t')$ or $\pi(t',r')$ before we encounter $\alpha$. Without loss of generality, we assume that we encounter $\pi(s',t')$. We assume that $s$, $t$, $r$ are ordered counterclockwise around the boundary of their geodesic hull (e.g., see Fig.~\ref{fig:subcase10lem}).
Consider the farthest Voronoi diagram $\fvd(s,t,r)$ of the three sites $s,t,r$ only (without considering other sites of $S$). Let $C(p)$ be the cell of $p\in \{s,t,r\}$ in the diagram.
As $d(c,s)=d(c,r)>d(c,t)$, $c$ belongs to the common boundary of $C(s)$ and $C(r)$, i.e., $c$ is on an edge of $\fvd(s,t,r)$. The point $\alpha$ divides $B(s,r)$ into two portions, one of which contains $c$. The above implies that the portion of $B(s,r)$ containing $c$ is an edge of $\fvd(s,t,r)$ (e.g., see Fig.~\ref{fig:inverseorder}). That edge partitions $\triangle(s,t,r)$ into two sides; one side contains $s'$ and the other contains $r'$. It is not difficult to see that the side containing $s'$ belongs to $C(r)$ while the other side belongs to $C(s)$.
Then, one can verify that the three cells $C(s)$, $C(t)$, and $C(r)$ in $\fvd(s,t,r)$ are ordered clockwise along the boundary of $P$ (e.g., see Fig.~\ref{fig:inverseorder}). According to Aronov~\cite{ref:AronovTh93}, $s$, $t$, and $r$ should also be ordered clockwise around the boundary of their geodesic hull. But this contradicts with the fact that $s$, $t$, and $r$ are ordered counterclockwise around the boundary of their geodesic hull.

\begin{figure}[t]
\begin{minipage}[t]{\textwidth}
\begin{center}
\includegraphics[height=2.5in]{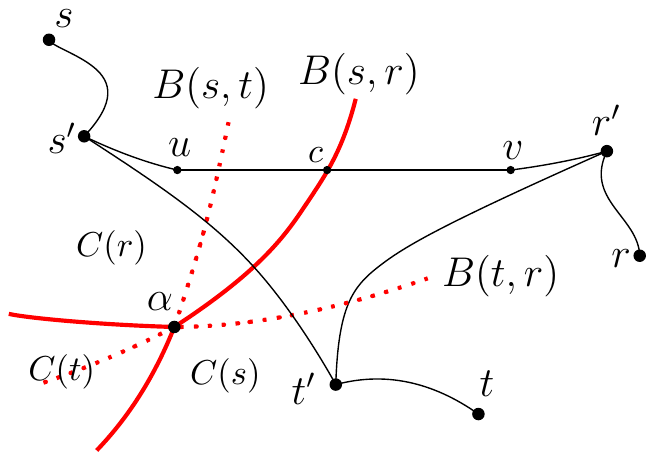}
\caption{\footnotesize Illustrating $\fvd(s,t,r)$, whose edges are depicted by thick (red) solid curves. The (red) dotted curves belong to bisectors but not on $\fvd(s,t,r)$. The point $\alpha$ is the only vertex of $\fvd(s,t,r)$ because it is equidistant from all three sites. The cells $C(s)$, $C(t)$, and $C(r)$ are ordered clockwise around $\alpha$, while $s$, $t$, and $r$ are ordered counterclockwise around the boundary of their geodesic hull, a contradiction.}
\label{fig:inverseorder}
\end{center}
\end{minipage}
\vspace{-0.15in}
\end{figure}

The above proves that $\alpha$ is in $P_1$. Oh and Ahn~\cite{ref:OhVo20} showed that
$\overline{\alpha p_t}$ intersects $\overline{uv}$. The main proof idea is that if this were not the case, then a vertex of $P$ would be on a bisector of two of the three sites $s$, $t$, and $r$, contradicting with the general position assumption (see the proof of Lemma 3.6~\cite{ref:OhVo20} for the detailed analysis). This leads to the lemma statement~(6), which further implies the lemma statement~(5).

To simplify the notation, let $q=\beta(s,t)$. By definition, $q$ is on the bisector $B(s,t)$. As $\alpha$ is equidistant from $s$, $t$, and $r$, $\alpha$ is also on $B(s,t)$. Let $p_{st}$ be the middle point of $\pi(s,t)$. Hence, $p_{st}\in B(s,t)$.

We claim that $\alpha$ must be on $B(s,t)$ between $p_{st}$ and $q$ (e.g., see Fig.~\ref{fig:subcase10}). Indeed, notice that $p_{st}$ is the point on $B(s,t)$ closest to $t$ and if we move a point $p$ from one end of $B(s,t)$ to the other end, $d(t,p)$ will first monotonically decrease until $p_{st}$ and then monotonically increase. By Observation~\ref{obser:monotone}, if we move a point $p$ along $B(s,t)$ from $q$ to $\alpha$, $d(t,p)$  will monotonically decrease. As such, $\alpha$ must be on $B(s,t)$ between $p_{st}$ and $q$.

We next argue that $q\in P_1$. Depending on whether $B(s,t)$ intersects $\overline{uv}$, there are two cases.
If $B(s,t)$ does not intersect $\overline{uv}$, then as $\alpha\in B(s,t)$ and $\alpha\in P_1$, $B(s,t)$ is in $P_1$. Since $q\in B(s,t)$, $q\in P_1$ holds.
If $B(s,t)$ intersects $\overline{uv}$, then since $\overline{uv}\subseteq \pi(s,r)$, by Observation~\ref{obser:old}, $B(s,t)$ intersects $\overline{uv}$ at a single point, denoted by $q_{st}$ (e.g., see Fig.~\ref{fig:subcase10}). To prove $q\in P_1$, since $\alpha\in P_1$ and $\alpha$ is on $B(s,t)$  between $q$ and $p_{st}$, it suffices to show that $p_{st}\in P_2$. Indeed, since $B(s,t)$ intersects $\overline{uv}$ at $q_{st}$ and $\overline{uv}\subseteq \pi(s,r)$, by Observation~\ref{obser:old}, $B(s,t)$ does not intersect any other point of $\pi(s,r)$. Hence, $B(s,t)$ does not intersect $\pi(s,s')\setminus\{s'\}$, which is a subpath of $\pi(s,r)$ and does not contain any point of $\overline{uv}$. This implies that $p_{st}$ cannot be on $\pi(s,s')\setminus\{s'\}$ and thus is on $\pi(s',t)$. Note that $s'$ is in $P_2$. Since both $s'$ and $t$ are in $P_2$, $\pi(s',t)$ is in $P_2$. As such, $p_{st}\in P_2$.

The above proves that $q\in P_1$ and $\alpha$ is on $B(s,t)$ between $q$ and $p_{st}$. In the following, we proceed to prove that $\alpha\in \triangle(s,r,q)$.

\begin{figure}[t]
\begin{minipage}[t]{\textwidth}
\begin{center}
\includegraphics[height=2.5in]{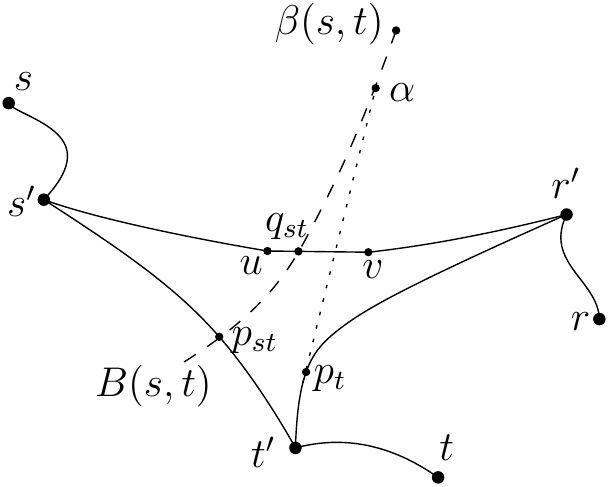}
\caption{\footnotesize Illustrating the relative positions of $\beta(s,t)$, $\alpha$, and $p_{st}$.}
\label{fig:subcase10}
\end{center}
\end{minipage}
\vspace{-0.15in}
\end{figure}

As $\alpha$ is on $B(s,t)$ between $q$ and $p_{st}$, $\alpha$ must be in the geodesic triangle $\triangle(s,t,q)$. Since $q\in B(s,t)$, due to the general position assumption, the incident edges of $q$ in $\pi(q,s)$ and $\pi(q,t)$ cannot be coincident~\cite{ref:AronovTh93}. This means that $q$ is the junction vertex of $\pi(q,s)$ and $\pi(q,t)$, and thus $q$ is an apex of $\triangle(s,t,q)$.
Since $t\in P_2$ and $q\in P_1$, $\pi(q,t)$ must cross $\overline{uv}$ at a point $p$; e.g., see Fig.~\ref{fig:subcase10lem}.
Since both $\alpha$ and $q$ are in $P_1$, $\alpha$ is also in the geodesic triangle $\triangle(s,p,q)$ and $q$ is an apex of $\triangle(s,p,q)$. Further, since $p\in \overline{uv}\subseteq \pi(s,r)$, $\triangle(s,p,q)$ is a subset of $\triangle(s,r,q)$ and $q$ is also an apex of $\triangle(s,r,q)$. As such, we obtain that $\alpha\in \triangle(s,r,q)$. This proves the lemma statement~(1).

The above also proves that $q$ is an apex of $\triangle(s,r,q)$. By definition, $u'$ and $v'$ are other two apexes of $\triangle(s,r,q)$. This proves the lemma statement~(2). Since $\alpha\in \triangle(s,r,q)$, the lemma statements (3) and (4) obviously hold.
\end{proof}

In light of Lemma~\ref{lem:30},
we can apply the tentative prune-and-search technique~\cite{ref:KirkpatrickTe95} on the three pseudo-convex chains specified in the lemma in a similar way as before to compute $\alpha$ in $O(\log n)$ time.

We summarize our algorithm for this case. First, we compute the
edge $\overline{uv}$, which can be done in $O(\log n)$ time using the
GH data structure by Lemma~\ref{lem:20}. Second, we compute the
junction vertex $t_{uv}$ of $\pi(t,u)$ and $\pi(t,v)$ in $O(\log n)$
time by the GH data structure~\cite{ref:GuibasOp89}. Third, we apply the tentative
prune-and-search technique on the three pseudo-convex chains as specified in
Lemma~\ref{lem:30}, along with the binary search scheme in
Lemma~\ref{lem:20} on the chains, to compute $\alpha$ in $O(\log n)$ time.

Recall that the above algorithm is based on the assumption that
$\alpha$ is a Voronoi vertex of $\fvd(S)$. However, when we invoke the procedure
during the geodesic sweeping algorithm
we do not know whether the assumption is true. Therefore, as a
final step, we add a validation procedure as follows. Suppose $\alpha$
is
the point returned by the algorithm. First, we check whether
$d(s,\alpha)=d(t,\alpha)=d(r,\alpha)$. If not, we return null.
Otherwise, we further check whether $d(c^*,\alpha)\leq
\min\{d(c^*,\beta(s,t)),d(c^*,\beta(t,r))\}$. This is because the Voronoi
vertex $\alpha$ is only useful if it is inside the current sweeping circle
$C$, whose geodesic distance to $c^*$ is at most
$\min\{d(c^*,\beta(s,t)),d(c^*,\beta(t,r))\}$ (because neither
$\beta(s,t)$ nor $\beta(t,r)$ is in the interior of $C$). Hence, if $d(c^*,\alpha)\leq
\min\{d(c^*,\beta(s,t)),d(c^*,\beta(t,r))\}$, then we return $\alpha$;
otherwise, we return null. This validation step
takes $O(\log n)$ time by the GH data structure.

\paragraph{At least one of $u$ and $v$ is not a polygon vertex.}
The above discusses the case where both $u$ and $v$ are polygon vertices. In the following, we consider the other case where at least one of them is not a polygon vertex, i.e., $u=s$ or/and $v=r$ (because all vertices of $\pi(s,r)$ except $s$ and $r$ are polygon vertices). In fact, this case is missed from the algorithm of Oh and Ahn~\cite{ref:OhVo20} (see the proof of Lemma~3.6~\cite{ref:OhVo20}). It turns out that Lemma~\ref{lem:30} still holds for this case and thus we can apply exactly the same algorithm as above. We prove the lemma below by reducing this case to the previous case where $u$ and $v$ are polygon vertices.

\begin{lemma}\label{lem:40}
%
%
%
%
%
Lemma~\ref{lem:30} still holds when $u=s$ or/and $v=r$.
\end{lemma}
\begin{proof}
Without loss of generality, we assume that $v$ is not a polygon vertex
and thus $v=r=r'$. We extend
$\overline{uv}$ in the direction from $u$ to $v$ until $\partial P$ at
a point $v'$ (e.g., see Fig.~\ref{fig:subcase100}).
If $u$ is also not a polygon vertex, then $u=s=s'$ and we extend
$\overline{uv}$ in the direction from $v$ to $u$ until $\partial P$ at
a point $u'$. If $u$ is a polygon vertex, we let $u'=u$.

Let $P'$ be the polygon obtained by adding $\overline{vv'}$ and
$\overline{uu'}$ to $P$. So $P'$ is a (weakly) simple polygon with $u$
and $v$ as two vertices. We claim that $B(s,t)$, $B(s,r)$, and
$B(t,r)$ are still the bisectors of $s$, $t$, and $r$ in $P'$. Before proving the claim, we proceed to prove the lemma with help of the claim. Due to the claim, since $u$ and $v$ are now both polygon vertices of $P'$, we can apply literally the same argument as in the previous case. Indeed, the argument only relies on the properties of the three bisectors, e.g., $\alpha$ is their common intersection. Now that the three bisectors do not change from $P$ to $P'$, the same argument still works. Thus, the lemma follows.

In the following, we prove the above claim. It is sufficient to show the following four properties: (1)
$\overline{vv'}\setminus\{v\}\cup \overline{uu'}\setminus\{u\}$ does not intersect any of the three bisectors $B(s,t)$, $B(t,r)$, and
$B(s,r)$; (2) $\overline{vv'}\setminus\{v\} \cup \overline{uu'}\setminus\{u\}$  does not intersect $\pi(s,p)$  for any point $p\in B(s,t)\cup B(s,r)$; (3) $\overline{vv'}\setminus\{v\} \cup \overline{uu'}\setminus\{u\}$ does not intersect $\pi(t,p)$  for any point $p\in B(s,t)\cup B(t,r)$;
(4) $\overline{vv'}\setminus\{v\} \cup \overline{uu'}\setminus\{u\}$ does not intersect $\pi(r,p)$ for any point $p\in B(s,r)\cup B(t,r)$. Below we will prove the above four properties only for $\overline{vv'}\setminus\{v\}$, as the proof for $\overline{uu'}\setminus\{u\}$ is similar.
We prove these properties in order.

\begin{figure}[t]
\begin{minipage}[t]{\textwidth}
\begin{center}
\includegraphics[height=2.5in]{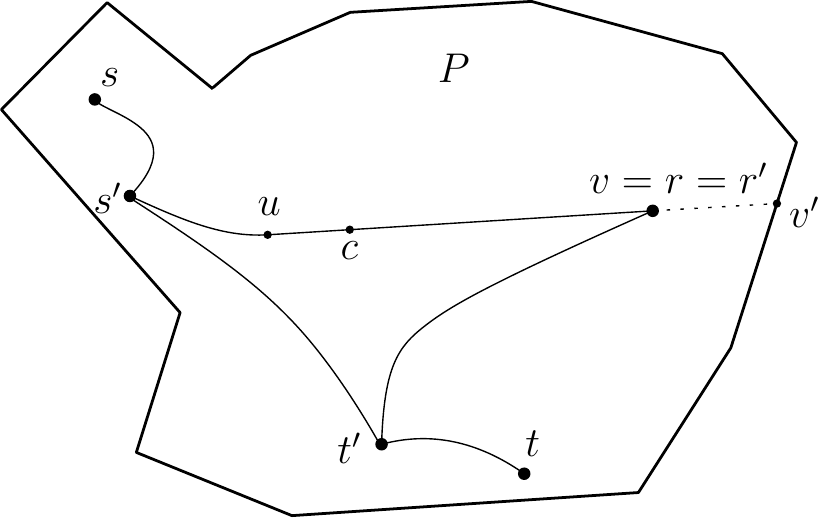}
\caption{\footnotesize Illustrating the proof of Lemma~\ref{lem:40}.}
\label{fig:subcase100}
\end{center}
\end{minipage}
\vspace{-0.15in}
\end{figure}

\paragraph{Property (1).}
First of all, since $v'$ is an extension of $\overline{uv}$, it holds that $\pi(s,v')=\pi(s,v)\cup \overline{vv'}$. Recall that $d(s,c)=d(c,r)>d(t,c)$ and $c\in \overline{uv}$.
Assume to the contrary that $\overline{vv'}$ intersects $B(s,t)$, say, at a point $z$. Then, $d(s,z)=d(s,c)+|\overline{cz}|$. On the other hand, by triangle inequality, $d(t,z)\leq d(t,c)+|\overline{cz}|$. Hence, we obtain $d(s,z)=d(s,c)+|\overline{cz}|>d(t,c)+|\overline{cz}|\geq d(t,z)$. However, since $z\in B(s,t)$, $d(s,z)=d(t,z)$, and thus contradiction occurs. This proves that  $\overline{vv'}$ does not intersect $B(s,t)$.

Because $\pi(s,v')$ contains $\overline{rv'}$, $d(s,p)>d(r,p)$ for any point $p\in \overline{vv'}$. Therefore, $\overline{vv'}\setminus\{v\}$ cannot intersect $B(s,r)$.

Next we prove the case for $B(t,r)$.
Let $a$ be the anchor of $r$ in $\pi(t,r)$. Since $d(t,c)<d(r,c)=|\overline{cr}|$, the angle $\angle(a,r,c)$ must be smaller than $\pi/2$, and thus the angle $\angle(v',r,a)$ is larger than $\pi/2$. Notice that $r$ is the junction vertex of $\pi(r,t)$ and $\pi(r,p)=\overline{rp}$ for any $p\in \overline{vv'}\setminus\{v\}$. Since the angle $\angle(p,r,a)=\angle(v',r,a)$ is larger than $\pi/2$, it must hold that $d(t,p)>d(r,p)$~\cite{ref:PollackCo89} (see Corollary 2). This implies that $p$ cannot be on $B(t,r)$. Thus, $\overline{vv'}\setminus\{v\}$ does not intersect $B(t,r)$.


This prove property (1).

\paragraph{Property (2).}
We now prove property (2). Let $p$ be any point in $B(s,t)\cup B(s,r)$. Assume to the contrary that $\pi(s,p)$ contains a point $p'\in \overline{vv'}\setminus{v}$. Then, since $\pi(s,p')$ contains $\pi(s,r)$, $\pi(s,p)$ contains $\pi(s,r)$ and thus contains $r$.

If $p\in B(s,r)$, then we immediately obtain contradiction as $\pi(s,p)$ cannot contain $r$.

Now consider the case $p\in B(s,t)$. Since $d(s,c)>d(t,c)$, there must be a point $p''\in \pi(s,c)$ such that $d(s,p'')=d(t,p'')$, i.e., $p''\in B(s,t)$. Since $\pi(s,p'')\subseteq \pi(s,c)$ and $c\neq r$, $\pi(s,p'')$ does not contain $r$.  If $p=p''$, we obtain that $\pi(s,p)$ does not contain $r$, which incurs contradiction. Hence, $p\neq p''$. Thus, $\pi(s,p)$ intersects $B(s,t)$ at two different points $p$ and $p''$. But this is not possible due to Observation~\ref{obser:old}.

This proves property (2).

\paragraph{Property (3).}
For property (3), let $p$ be any point in $B(s,t)\cup B(t,r)$. Assume to the contrary that $\pi(t,p)$ contains a point $p'\in \overline{vv'}\setminus\{v\}$.

We first discuss the case $p\in B(s,t)$. Consider the geodesic
triangle $\triangle(s,t,p)$; e.g., see Fig.~\ref{fig:propterty30}. Since $p\in B(s,t)$, $p$ must be an apex
of $\triangle(s,t,p)$. Let $a$ be the junction vertex of $\pi(s,p)$
and $\pi(s,t)$ and let $b$ the junction vertex of $\pi(t,s)$ and
$\pi(t,p)$. Hence, $a$ and $b$ are two apexes of $\triangle(s,t,p)$.
By a similar argument as Observation~\ref{obser:00}, the middle point
$p_{st}$ of $\pi(s,t)$ must be on $\pi(a,b)$, i.e., the side of
$\triangle(s,t,p)$ opposite to $p$. Hence, the portion of $B(s,t)$
between $p$ and $p_{st}$, denoted by $B$, separates $\triangle(s,t,p)$
into two parts. As $\pi(t,p)=\pi(t,b)\cup \pi(b,p)$, $p'$ is either in
$\pi(t,b)$ or in $\pi(b,p)$.

\begin{figure}[t]
\begin{minipage}[t]{\textwidth}
\begin{center}
\includegraphics[height=1.8in]{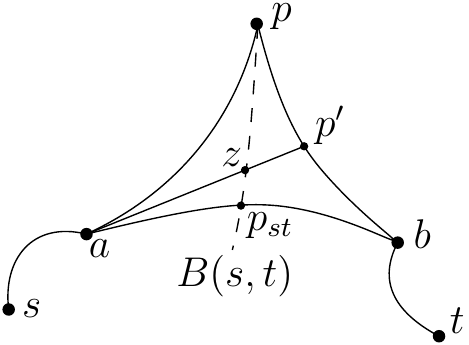}
\caption{\footnotesize Illustrating the pseudo-triangle $\triangle(s,t,p)$.}
\label{fig:propterty30}
\end{center}
\end{minipage}
\vspace{-0.15in}
\end{figure}

\begin{itemize}
\item
If $p'$ is in $\pi(t,b)$, then $p'$ is in $\pi(s,t)$ as $\pi(t,b)$ is a subpath of $\pi(s,t)$. Therefore, $\pi(s,p')$ is a subpath of $\pi(s,t)$. Recall that $\pi(s,r)\subseteq \pi(s,p')$. We thus obtain that $\pi(s,r)$ is a subpath of $\pi(s,t)$. Since $c\in \pi(s,r)$, we obtain that $s$, $c$, $r$, and $t$ are all on $\pi(s,t)$ in this order. Hence, $d(c,r)\leq d(c,t)$, which incurs contradiction as $d(c,r)>d(c,t)$.

\item
If $p'$ is in $\pi(b,p)$, then $\pi(s,p')$ must intersect $B$, say, at a point $z$ (e.g., see Fig.~\ref{fig:propterty30}). By Lemma~\ref{lem:sp}, it holds that $d(s,z)\geq d(z,p')$, and thus $d(s,z)\geq d(s,p')/2$. Since $r\in \pi(s,p')$, we have $d(s,z)\geq d(s,r)/2=d(s,c)$. This implies that $c\in \pi(s,z)$. Because $z\in B$, we have $d(s,c)\leq d(t,c)$. But this contradicts with the fact $d(s,c)>d(t,c)$.
\end{itemize}

The above obtains contradiction for the case $p\in B(s,t)$.

We next discuss the case $p\in B(t,r)$. The bisector $B(t,r)$ divides
$P$ into two subpolygons; let $P_t$ be the one containing $t$ and let $P_r$ denote the one containing $r$. We
assume that neither $P_s$ nor $P_r$ contains $B(t,r)$. As
$p\in B(t,r)$, the entire path $\pi(t,p)$ is in $P_t\cup B(s,t)$.
Since $p'\in \overline{vv'}\setminus\{v'\}$, by a similar argument using the angle at $r$ as that
for property (1), we can show that $d(t,p')>d(r,p')$. This implies
that $p'$ is in $P_s$. Therefore, $p'$ cannot be in $\pi(s,p)$, a
contradiction.

This proves property (3).

\paragraph{Property (4).}
For property (4), assume to the contrary that $\pi(r,p)$
contains a point $p'\in \overline{vv'}\setminus\{v\}$.

We first discuss the case $p\in B(t,r)$. Let $a$ be the anchor of $r$ in $\pi(t,r)$. Recall that we have shown before that the angle $\angle (a,r,v')$ is larger than $\pi/2$. Consider the geodesic triangle $\triangle(t,r,p)$.

\begin{itemize}
\item
If $r$ is not an apex of $\triangle(t,r,p)$, then $\overline{ra}\in \pi(r,t)\cap \pi(r,p)$. Since $p'\in \pi(r,p)$ and $p'\not\in \overline{ra}$, we obtain that $\pi(r,p')$ contains $a$. However, since $\angle (a,r,v') > \pi/2$ and $p'\in \overline{rv'}$, $\pi(r,p')=\overline{rp'}$ does not contain $a$, a contradiction.

\item
If $r$ is an apex of $\triangle(t,r,p)$, then since $p\in B(t,r)$, the angle $\angle (a,r,b)$ must be smaller than $\pi/2$, where $b$ is the anchor of $r$ in $\pi(p,r)$. As $p'\in \pi(r,p)$ and $\pi(r,p')=\overline{rp'}$, we obtain that $p'\in \overline{rb}$ and thus $\angle (a,r,b)=\angle (a,r,p')$.
Hence, $\angle (a,r,p')$ is smaller than $\pi/2$. However,
$\angle (a,r,p')=\angle (a,r,v')$, which is larger than $\pi/2$. Thus we obtain contradiction.
\end{itemize}

We then discuss the case $p\in B(s,r)$. Note that $u$ is the anchor of $r$ in $\pi(s,r)$. Hence, the angle $\angle(u,r,v')$ is equal to $\pi$, which is larger than $\pi/2$. Consequently, we can follow the same analysis as above to obtain contradiction.

This proves property (4).

The lemma thus follows.
\end{proof}

We finally prove the following technical lemma, which is needed in the proof of Lemma~\ref{lem:40}. The lemma, which establishes a very basic property of shortest paths in simple polygons, may be interesting in its own right.

\begin{lemma}\label{lem:sp}
Let $s$ and $t$ be any two points in $P$ such that $B(s,t)$ does not contain any vertex of $P$.
Suppose $p$ is a point in $B(s,t)$ and $p'$ is a point in $\pi(t,p)$.
Then, $\pi(s,p')$ intersects $B(s,t)$ at a single point $z$ and $d(s,z)\geq d(z,p')$ (in particular, $d(s,z)> d(z,p')$ if $p'\neq t$); e.g., see Fig.~\ref{fig:propterty30}.
\end{lemma}
\begin{proof}
We first consider a special case where $p$ is the middle point $p_{st}$ of $\pi(s,t)$. Note that $p_{st}\in B(s,t)$. In this case, $z=p_{st}$ and $\pi(s,t)=\pi(s,z)\cup \pi(z,t)$. Hence,
$\pi(s,p')=\pi(s,z)\cup \pi(z,p')$ and $d(s,z)=d(z,t)\geq d(z,p')$ (and $d(s,z)=d(s,t)>d(z,p')$ if $p'\neq t$).

In the following we assume $p\neq p_{st}$.
Consider the geodesic triangle $\triangle(s,t,p)$. Since $p\in B(s,t)$, $p$ is an apex of $\triangle(s,t,p)$. Let $a$ be the junction vertex of $\pi(s,t)$ and $\pi(s,p)$ and $b$ be the junction vertex of $\pi(t,s)$ and $\pi(t,p)$ (e.g., see Fig.~\ref{fig:propterty30}). Hence, $a,b,p$ are the three apexes of $\triangle(s,t,p)$. In the following discussion we will use $\triangle(a,b,p)$ instead. By a similar argument as Observation~\ref{obser:00}, $p_{st}$ is in $\pi(a,b)$.

The bisector $B(s,t)$ partitions $P$ into two connected subpolygons $P_s$ and $P_t$ such that $P_s$ contains $s$ and $P_t$ contains $t$~\cite{ref:AronovOn89}. We assume that neither $P_s$ nor $P_t$ contains $B(s,t)$.
Since $p\in B(s,t)$, $\pi(s,p)\setminus\{p\}\subseteq P_s$ and $\pi(t,p)\setminus\{p\}\subseteq P_t$~\cite{ref:AronovOn89}. This implies that $\pi(s,p')$ intersects $B(s,t)$ at a single point $z$ because $s\in P_s$ and $p'\not\in P_s$.

Since $p'\in \pi(t,p)$, $p'$ is either in $\pi(t,b)$ or in $\pi(b,p)\setminus\{b\}$. In the former case, $\pi(s,b)\subseteq \pi(s,p')\subseteq \pi(s,t)$ and thus $z=p_{st}$. Hence, $d(s,z)=d(s,t)/2$. On the other hand, $\pi(z,p')\subseteq\pi(z,t)$ and thus $d(z,p')\leq d(z,t)=d(s,t)/2$ (and $d(z,p')<d(z,t)$ if $p'\neq t$). Therefore, we obtain that $d(s,z)\geq d(z,p')$ (and $d(s,z) > d(z,p')$ if $p'\neq t$).

It remains to discuss the case $p'\in \pi(b,p)\setminus\{b\}$. In this case, $p'$ is a point on the boundary of $\triangle(a,b,p)$. Let $q$ be the anchor of $p'$ in $\pi(s,p')$.  Because $p'$ is on the boundary of $\triangle(a,b,p)$, $\pi(s,a)\subseteq \pi(s,p')$ and $q$ is on the boundary of $\triangle(a,b,p)$. Since $p'$ is in $\pi(b,p)$, which is the opposite side of the apex $a$, $q$ cannot be on $\pi(b,p)$ and thus must be on $\pi(a,p)\cup \pi(a,b)$. Note that $\pi(a,b)=\pi(a,p_{st})\cup \pi(p_{st},b)$.

Our goal is to prove $d(s,z)> d(z,p')$.
Assume to the contrary that $d(s,z)\leq d(z,p')$. In the following we will obtain $d(s,p)>d(t,p)$, which incurs contradiction as $p\in B(s,t)$.
Depending on whether $q$ is in $\pi(a,p)\cup \pi(a,p_{st})$ and $\pi(p_{st},b)\setminus\{p_{st}\}$, there are two cases. We first show that the latter case can be reduced to the former case.

\begin{figure}[t]
\begin{minipage}[t]{\textwidth}
\begin{center}
\includegraphics[height=1.8in]{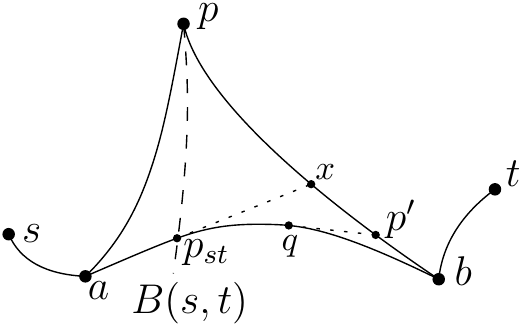}
\caption{\footnotesize Illustrating the proof of Lemma~\ref{lem:sp} for the case $q\in \pi(p_{st},b)\setminus\{p_{st}\}$.}
\label{fig:sp}
\end{center}
\end{minipage}
\vspace{-0.15in}
\end{figure}

If $q\in \pi(p_{st},b)\setminus\{p_{st}\}$, then $z=p_{st}$; e.g., see Fig.~\ref{fig:sp}. Hence, $d(z,b)\leq d(t,z)=d(s,z)\leq d(z,p')$.
If we move a point $x$ on $\pi(b,p)$ from $b$ to $p$, then $d(s,x)$ is strictly convex~\cite{ref:PollackCo89} (see the proof of Lemma~1), and more precisely, $d(s,x)$ first strictly decreases until a point $x^*$ and then strictly increases. We claim that $x^*\in \pi(b,p')$. Indeed, assume to the contrary that $x^*\in \pi(p',p)$. Then, we can obtain $d(s,p')<d(s,b)$. Because both $\pi(s,p')$ and $\pi(s,b)$ contain $z$, we derive that $d(z,p')<d(z,b)$, which contradicts with $d(z,p')\geq d(z,b)$. Due to the above claim, as we move $x$ from $p'$ to $p$ along $\pi(p',p)$, $\pi(s,x)$ will strictly increase. We stop moving $x$ when $\overline{xx'}$ contains $p_{st}$ (e.g., see Fig.~\ref{fig:sp}), where $x'$ is the anchor of $x$ in $\pi(s,x)$. Note that such a moment must exist as $\pi(s,p)$ does not contain $p_{st}$. As $p_{st}\in \pi(s,p')$, $z$ is still $p_{st}$. Since $d(s,x)>d(s,p')$ and both $\pi(s,x)$ and $\pi(s,p')$ contain $z$, we can obtain that $d(z,p')<d(z,x)$. As $d(s,z)\leq d(z,p')$, we deduce that $d(s,z)\leq d(z,x)$. Now we obtain an instance of the first case (i.e., $q\in \pi(a,p)\cup \pi(a,p_{st})$) because $x'\in \pi(a,p_{st})$, i.e., we can consider $x$ as a new point $p'$ and obtain contradiction by using the analysis given below for the first case.

\begin{figure}[t]
\begin{minipage}[t]{\textwidth}
\begin{center}
\includegraphics[height=2.0in]{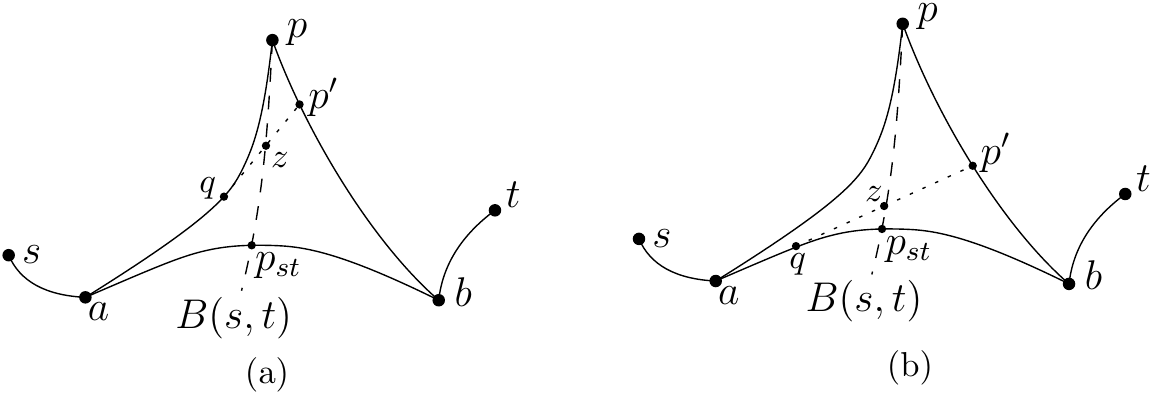}
\caption{\footnotesize Illustrating the proof of Lemma~\ref{lem:sp}: (a) $q\in \pi(a,p)$; (b) $q\in  \pi(a,p_{st})$.}
\label{fig:sp00}
\end{center}
\end{minipage}
\vspace{-0.15in}
\end{figure}

In the following we consider the case where $q\in \pi(a,p)\cup \pi(a,p_{st})$; e.g., see Fig.~\ref{fig:sp00}. In this case, $\overline{qp'}$ intersects $B(s,t)$ at $z$ and thus $d(z,p')=\overline{zp'}$. There are two subcases depending on whether $q\in \pi(a,p)$ or $q\in \pi(a,p_{st})$. For each case, we will construct a geodesic triangle $\triangle(a',b',p)$ (which may not be in $P$) with the following properties: (1) $a'$, $b'$, and $p$ are its three apexes; (2) the length of the side of $\triangle(a',b',p)$ connecting $a'$ and $p$, denoted by $l(a',p)$, is at most $d(s,p)$; (3) the length of the side of $\triangle(a',b',p)$ connecting $b'$ and $p$, denoted by $l(b',p)$, is equal to $d(b,p)$; (4) the angle at $b'$ formed by its two incident edges of $\triangle(a',b',p)$ is at least $\pi/2$. These properties together lead to $d(s,p)>d(t,p)$. Indeed, due to property (4), it holds that $l(b',p)<l(a',p)$~\cite{ref:PollackCo89} (see Corollary 2). Combining with properties (2) and (3), we have $d(t,p)=l(b',p)<l(a',p)\leq d(s,p)$. This incurs contradiction since $d(t,p)=d(s,p)$.

\begin{figure}[t]
\begin{minipage}[t]{\textwidth}
\begin{center}
\includegraphics[height=2.0in]{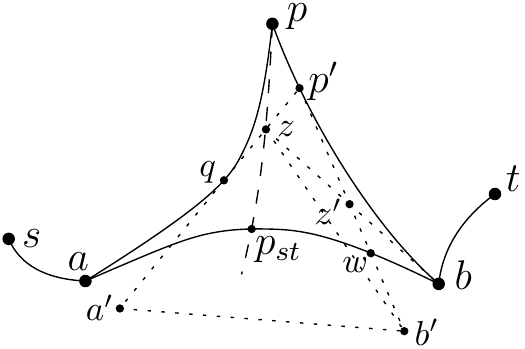}
\caption{\footnotesize Illustrating the definitions of $a'$, $b'$, $w$, and $z'$ in the proof of Lemma~\ref{lem:sp}.}
\label{fig:sp10}
\end{center}
\end{minipage}
\vspace{-0.15in}
\end{figure}

\paragraph{The first subcase $q\in \pi(a,p)$.}
We begin with the subcase $q\in \pi(a,p)$; e.g., see Fig.~\ref{fig:sp00}(a). We extend $\overline{p'q}$ along the direction from $p'$ to $q$ until a point $a'$ such that $|\overline{qa'}|=d(s,q)$; note that $\overline{qa'}$ may not be in $P$. Refer to Fig.~\ref{fig:sp10}. Note that $\pi(p,q)\cup \overline{qa'}$ is still a convex chain and its length $l(a',p)$ is equal to $|\overline{qa'}|+d(q,p)=d(s,q)+d(q,p)=d(s,p)$. Let $p''$ be the vertex of $\pi(p',p)$ incident to $p'$. We extend $\overline{p''p'}$ along the direction from $p''$ to $p'$ until a point $b'$ such that $|\overline{p'b'}|=d(t,p')$. Note that $\pi(p,p')\cup \overline{p'b'}$ is still a convex chain and its length $l(b',p)$ is equal to $|\overline{p'b'}|+d(p',p)=d(t,p')+d(p',p)=d(t,p)$. In the following we show that the angle $\angle(p',b',a')$ is at least $\pi/2$; this will prove all four properties described above for the geodesic triangle $\triangle(a',b',p)$.

We claim that $|\overline{zb'}|\leq d(z,t)$. Before proving the claim, we first show $\angle(p',b',a')\geq \pi/2$ by using the claim. Indeed, notice that $|\overline{za'}|=|\overline{zq}|+|\overline{qa'}|=|\overline{zq}|+d(s,q)=d(s,z)$. As $z\in B(s,t)$, $d(s,z)=d(t,z)$. Hence, $|\overline{za'}|\geq |\overline{zb'}|$. Recall that $d(s,z)\leq d(z,p')=|\overline{zp'}|$. Therefore, we have $|\overline{zp'}|\geq |\overline{za'}|\geq |\overline{zb'}|$.  If we draw a circle centered at $z$ with radius equal to $|\overline{za'}|$, then $a'$ is on the circle, $b'$ is inside or on the circle, and $p'$ is outside or on the circle. Since $\overline{a'p'}$ contains a diameter of the circle, we obtain that $\angle(p',b',a')\geq \pi/2$.

We proceed to prove the claim $|\overline{zb'}|\leq d(z,t)$.
Notice that $\overline{p'b'}$ must intersect $\pi(z,b)$. Indeed, unless $\overline{p'b'}$ contains $b$ (in which case it is vacuously true that $\overline{p'b'}$ intersects $\pi(z,b)$), as $\overline{p'b'}$ is tangent to $\pi(b,p)$ at $p'$, if we move on $\overline{p'b'}$ from $p'$ towards $b'$, we will enter the interior of $\triangle(a,b,p)$ and let $w$ be the first point on the boundary of $\triangle(a,b,p)$ we meet during the above movement after $p'$ (i.e., we will go outside $\triangle(a,b,p)$ after $w$; e.g., see Fig.~\ref{fig:sp10}). Then, $\overline{p'w}$ separates $z$ and $b$ on its two sides, and thus $\overline{p'w}$ must intersect $\pi(z,b)$, say, at a point $z'$.
We next prove $|\overline{z'b'}|\leq d(z',t)$. Indeed, $d(t,p')=|\overline{p'b'}|=|\overline{p'z'}|+|\overline{z'b'}|$. On the other hand, by triangle inequality, $d(t,p')\leq |\overline{p'z'}|+d(z',t)$. Hence, we obtain $|\overline{z'b'}|\leq d(z',t)$. Consequently, by triangle inequality, $|\overline{zb'}|\leq d(z,z')+|\overline{z'b'}|\leq d(z,z')+d(z',t)=d(z,t)$. This proves the claim.

The above proves the first subcase $q\in \pi(a,p)$.

\paragraph{The second subcase $q\in \pi(a,p_{st})$.}
We next discuss the subcase $q\in \pi(a,p_{st})$; e.g., see Fig.~\ref{fig:sp00}(b). The analysis is somewhat similar. First of all, we define $b'$ in the same way as above; by the same argument, we have: (1) $\pi(p,p')\cup \overline{p'b'}$ is a convex chain and its length is equal to $d(p,b)$; (2) $|\overline{zb'}|\leq d(z,t)$.

\begin{figure}[t]
\begin{minipage}[t]{\textwidth}
\begin{center}
\includegraphics[height=2.0in]{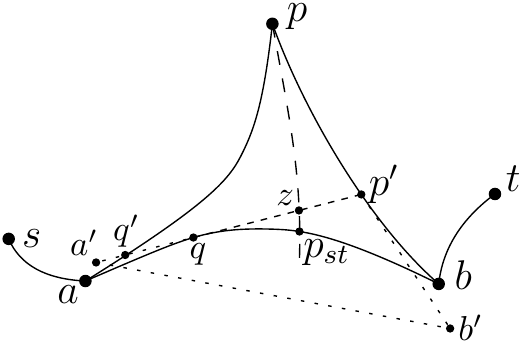}
\caption{\footnotesize Illustrating the definitions of $a'$, $b'$, and $q'$ in the proof of Lemma~\ref{lem:sp}.}
\label{fig:sp20}
\end{center}
\end{minipage}
\vspace{-0.15in}
\end{figure}

The point $a'$ is now defined in a slightly different way. Refer to Fig.~\ref{fig:sp20}. We extend $\overline{p'q}$ from $p'$ to $q$ until a point $a'$ that satisfies the following two conditions: (1) $|\overline{qa'}|\geq d(s,q)$; (2) $\overline{p'a'}$ intersects $\pi(a,p)$, say, at a point $q'$ (this is possible as $\overline{p'q}$ is tangent to $\pi(a,p_{st})$ at $q$). Note that $a'=q'$ if and only if $|\overline{qq'}|\geq d(s,q)$. Also note that $|\overline{qa'}|=d(s,q)$ if $a'\neq q'$. Next we prove the four properties of the geodesic triangle $\triangle(a',b',p)$.

First of all, notice that $|\overline{a'z}|=|\overline{a'q}|+|\overline{qz}|\geq d(s,q)+|\overline{qz}|=d(s,z)$. Recall that $|\overline{zb'}|\leq d(t,z)=d(s,z)$ and $d(s,z)\leq |\overline{zp'}|$. By a similar argument as before for the first subcase, the angle $\angle(a',b',p')$ is at least $\pi/2$. By the definitions of $a'$ and $q'$, $\pi(p,q')\cup \overline{q'a'}$ is a convex chain and we will show below that its length $l(a',p)$ is at most $d(s,p)$, which will prove all four properties of $\triangle(a',b',p)$.

To prove $l(a',p)\leq d(s,p)$, as $l(a',p)=|\overline{a'q'}| + d(q',p)$ and $d(s,p)=d(s,q')+d(q',p)$, it is sufficient to prove $|\overline{a'q'}|\leq d(s,q')$. If $a'=q'$, this is obviously true. We thus assume $a'\neq q'$. Then, $|\overline{qq'}|< |\overline{qa'}| = d(s,q)$. Hence, $d(s,p')=d(s,q)+|\overline{qp'}|=|\overline{qa'}| + |\overline{qp'}|=|\overline{a'p'}|=|\overline{a'q'}|+|\overline{q'p'}|$. On the other hand, by triangle inequality, $d(s,p')\leq d(s,q')+|\overline{q'p'}|$. Therefore, we obtain $|\overline{a'q'}|\leq d(s,q')$.

The above proves the second subcase $q\in \pi(a,p_{st})$.

The lemma thus follows.
\end{proof}


\vspace{-0.2in}
\subsubsection{The case $c\in \pi(s,s')\setminus\{s'\}$}

We now consider the case $c\in \pi(s,s')\setminus\{s'\}$. For
this case, Oh and Ahn~\cite{ref:OhVo20} (see the proof of
Lemma~3.6~\cite{ref:OhVo20}) claimed that $\alpha$ does not exist. However, this
is not correct; see Fig.~\ref{fig:counterexample} for a
counterexample.

\begin{figure}[h]
\begin{minipage}[t]{\textwidth}
\begin{center}
\includegraphics[height=1.8in]{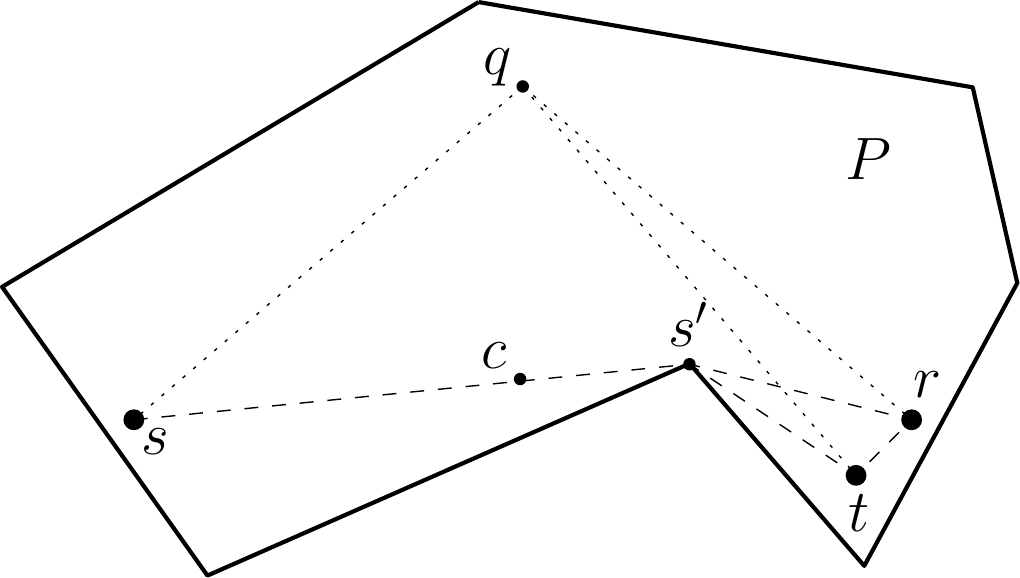}
\caption{\footnotesize Illustrating an example where $\alpha$ exists when $c\in \pi(s,s')\setminus\{s'\}$. The apexes of
the geodesic triangle $\triangle(s,r,t)$ are $s'$, $r'=r$, and $t'=t$.
$|\overline{s't}|<|\overline{s'r}|$. $c$ is the middle point of
$\pi(s,r)=\overline{ss'}\cup \overline{s'r}$. However, one can verify (e.g.,
by a ruler) that $q$ is equidistant to $s$, $r$, and $t$, and
thus $\alpha=q$ exists.}
\label{fig:counterexample}
\end{center}
\end{minipage}
\vspace{-0.1in}
\end{figure}

Let $v$ be the vertex incident to $s'$ in $\pi(s',r')$. To make the notation consistent with the
previous subcase, we let $u=s'$.
We have the following lemma (e.g., see Fig.~\ref{fig:subcase20new}), which is literally the same as  Lemma~\ref{lem:30} (the proof is also somewhat similar).

\begin{figure}[t]
\begin{minipage}[t]{\textwidth}
\begin{center}
\includegraphics[height=2.5in]{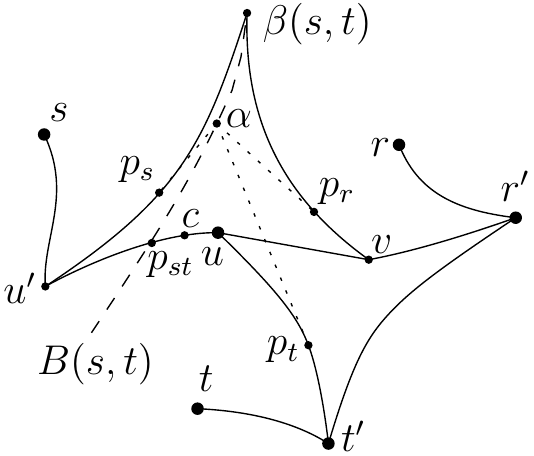}
\caption{\footnotesize Illustrating Lemma~\ref{lem:50}. In this example, $v'=v$.}
\label{fig:subcase20new}
\end{center}
\end{minipage}
\vspace{-0.15in}
\end{figure}

\begin{lemma}\label{lem:50}
\begin{enumerate}
\item
$\alpha$ must be in the geodesic triangle $\triangle(s,r,\beta(s,t))$.

\item
The apexes of $\triangle(s,r,\beta(s,t))$ are $u'$, $v'$, and $\beta(s,t)$, where $u'$ (resp., $v'$) is the junction vertex of and $\pi(s,r)$ and $\pi(s,\beta(s,t))$ (resp., $\pi(r,\beta(s,t))$) (in Fig.~\ref{fig:subcase20new}, $v'=v$).

\item
$p_s$ must be on the pseudo-convex chain $\pi(u',\beta(s,t))\cup \pi(u',v')$ and $\overline{\alpha p_s}$ is
tangent to the chain.

\item
$p_r$ must be on the pseudo-convex chain $\pi(v',\beta(s,t))\cup\pi(v',u')$ and $\overline{\alpha p_r}$ is tangent to the chain.

\item
$p_t$ must be on the pseudo-convex chain $\pi(t_{uv},u)\cup \pi(t_{uv},v)$ and $\overline{\alpha p_t}$ is
tangent to the chain, where $t_{uv}$ is the junction vertex of $\pi(t,u)$ and $\pi(t,v)$.

\item
$\overline{\alpha p_t}$ intersects $\overline{uv}$.
\end{enumerate}
\end{lemma}
\begin{proof}
As $c\in \pi(s,s')\setminus\{s'\}$ and $u=s'$, $u$ cannot be $s$ and thus must be a polygon vertex. But $v$ can be either a polygon vertex or the site $r$. We assume that $v$ is a polygon vertex since the other case can be reduced to this case by the same  technique as in the proof of Lemma~\ref{lem:40}.

As both $u$ and $v$ are polygon vertices, $\overline{uv}$ divides $P$ into two sub-polygons; one of them, denoted by $P_1$, does not contain $t$ and we use $P_2$ to denote the other one. Let $P'$ be the one of $P_1$ and $P_2$ that contains $s$. We will argue later that $P'$ must be $P_1$.

We first show that the bisector $B(s,t)$ is in $P'$. Recall that $p_{st}$
is the middle point of $B(s,t)$. Since $c\in \pi(s,s')\setminus\{s'\}$ and
$d(s,c)=d(r,c)>d(c,t)$, $d(s,t)=d(s,c)+d(t,c)<d(s,c)+d(c,r)=d(s,r)$. Hence $d(s,p_{st})=d(s,t)/2<d(s,r)/2=d(s,c)$. Thus, $p_{st}\in \pi(s,c)\setminus\{c\}$. Note that $\pi(s,c)$ is in $P'$ since $\pi(s,c)\subseteq \pi(s,u)$ and $\pi(s,u)\in P'$ (the latter holds because
both $s$ and $u$ are in $P'$). Therefore, $p_{st}\in P'$. Further, as $p_{st}\in \pi(s,c)\subseteq \pi(s,u)\setminus\{u\}$, $p_{st}\not\in \overline{uv}$.
Since $p_{st}\in \pi(s,r)$, $B(s,t)$ does not intersect $\pi(s,r)$ other than $p_{st}$ by Observation~\ref{obser:old}. As $\overline{uv}\subseteq \pi(s,r)$, we obtain that $B(s,t)$ does not intersect $\overline{uv}$. Because $p_{st}$ is in $B(s,t)\cap P'$ and $B(s,t)$ does not intersect
$\overline{uv}$, $B(s,t)$ must be in $P'$.

Since $c\in \pi(s,s')\setminus\{s'\}$ and $c\in B(s,r)$, $B(s,r)$ is also in $P'$ by following the analysis similar to the above. As $\alpha$ is equidistant from $s$, $t$, and $r$, $\alpha$ is on both $B(s,t)$ and $B(s,r)$. Therefore, $\alpha$ is in $P'$.

\begin{figure}[t]
\begin{minipage}[t]{\textwidth}
\begin{center}
\includegraphics[height=2.5in]{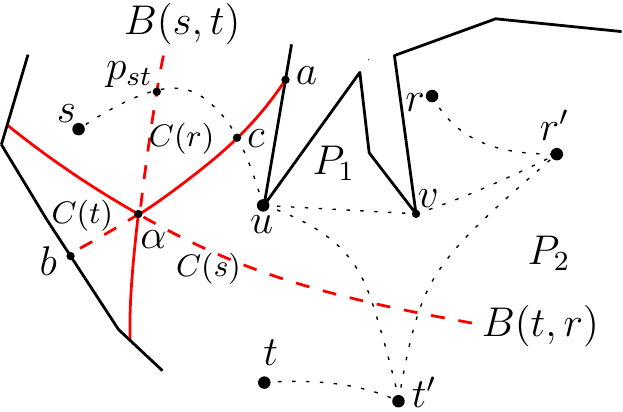}
\caption{\footnotesize Illustrating $\fvd(s,t,r)$, whose edges are depicted by thick (red) solid curves. The black solid segments are part of the boundary of $P$. The (red) dotted curves belong to bisectors but not on $\fvd(s,t,r)$. The red curve connecting $a$ and $b$ is $B(s,r)$; the portion between $c$ and $a$ (resp., $b$) is $B_1(s,r)$ (resp., $B_2(s,r)$).
The point $\alpha$ is the only vertex of $\fvd(s,t,r)$ because it is equidistant from all three sites. The cells $C(s)$, $C(t)$, and $C(r)$ are ordered clockwise around $\alpha$, while $s$, $t$, and $r$ are ordered counterclockwise around the boundary of their geodesic hull, a contradiction.}
\label{fig:inverseorder10}
\end{center}
\end{minipage}
\vspace{-0.15in}
\end{figure}

We next argue that $P'$ must be $P_1$. We assume that $s$, $t$, $r$ are ordered counterclockwise around the boundary of their geodesic hull (e.g., see Fig.~\ref{fig:subcase20new}). Assume to the contrary that $P'$ is $P_2$ (e.g., see Fig.~\ref{fig:inverseorder10}). The point $c$ divides $B(s,r)$ into two portions, one going above $\pi(s,r)$ and the other going below $\pi(s,r)$ (we intuitively assume that $\pi(s,r)$ from $s$ to $r$ goes ``horizontally'' from left to right); let $B_1(s,r)$ (resp., $B_2(s,r)$) be the first (resp., second) portion (e.g., see Fig.~\ref{fig:inverseorder10}, where the red curve between $c$ and $a$ is $B_1(s,r)$ and the red curve between $c$ and $b$ is $B_2(s,r)$). Since $P'$ is $P_2$, the two polygon edges of $P$ incident to $u$ must be from the ``above'' of $\pi(s,r)$. Hence, for any point $p\in B_1(s,r)$, both shortest paths $\pi(s,p)$ and $\pi(t,p)$ must contain $u$. Thus, $d(r,p)=d(r,u)+d(u,p)$ and $d(t,p)=d(t,u)+d(u,p)$. Since $d(r,c)>d(t,c)$ and both $\pi(r,c)$ and $\pi(t,c)$ contain $u$, $d(r,u)>d(t,u)$ holds. Therefore, $d(r,p)>d(t,p)$. This implies that no point on $B_1(s,r)$ is equidistant from $r$ and $t$, and thus $B_1(s,r)$ does not contain $\alpha$. As $\alpha\in B(s,r)$, we have $\alpha\in B_2(s,r)$. Consider the farthest Voronoi diagram $\fvd(s,t,r)$ of the three sites $s,t,r$ only (without considering other sites of $S$). Let $C(p)$ be the cell of $p\in \{s,t,r\}$ in the diagram.
As $d(c,s)=d(c,r)>d(c,t)$, $c$ belongs to the common boundary of $C(s)$ and $C(r)$, i.e., $c$ is on an edge of $\fvd(s,t,r)$. The point $\alpha$ divides $B(s,r)$ into two portions, one of which contains $c$. The above implies that the portion of $B(s,r)$ containing $c$ is an edge of $\fvd(s,t,r)$ (e.g., see Fig.~\ref{fig:inverseorder10}).
Recall that $p_{st}\in \pi(s,c)\setminus\{c\}$. Since $\pi(s,c)\subset\pi(s,r)$ and $c$ is the middle point of $\pi(s,r)$, we obtain that $d(s,p_{st})=d(t,p_{st})<d(r,p_{st})$, implying that $p_{st}\in C(r)$. The point $\alpha$ partitions $B(s,t)$ into two portions, one of which contains $p_{st}$; since $p_{st}\in C(r)$, the portion of $B(s,t)$ containing $p_{st}$ is not an edge of $F(s,t,r)$.
Then, one can verify that the three cells $C(s)$, $C(t)$, and $C(r)$ in $\fvd(s,t,r)$ are ordered clockwise along the boundary of $P$ (e.g., see Fig.~\ref{fig:inverseorder10}). According to Aronov~\cite{ref:AronovTh93}, $s$, $t$, and $r$ should also be ordered clockwise around the boundary of their geodesic hull. But this contradicts with the fact that $s$, $t$, and $r$ are ordered counterclockwise around the boundary of their geodesic hull.

The above proves that $P'$ is $P_1$. Since $B(s,t)\in P'$ and $\alpha\in B(s,t)$, we obtain that $\alpha\in P_1$.

We next argue that $\alpha$ must be in the geodesic triangle $\triangle(s,r,\beta(s,t))$. The argument is similar to the proof of Lemma~\ref{lem:30}, so we briefly discuss it. To simplify the notation, let $q=\beta(s,t)$. Since $q\in B(s,t)$ and $B(s,t)\in P_1$, $q$ is in $P_1$. By the same analysis as in the proof of Lemma~\ref{lem:30}, $\alpha$ is on $B(s,t)$ between $p_{st}$ and $q$ (e.g., see Fig.~\ref{fig:subcase20new}), and thus $\alpha$ is in the geodesic triangle $\triangle(s,t,q)$ and $q$ is an apex of $\triangle(s,t,q)$. Since $t\in P_2$ and $q\in P_1$, $\pi(q,t)$ must cross $\overline{uv}$ at a point $p$, and thus $\alpha$ is also in $\triangle(s,p,q)$ and $q$ is an apex of $\triangle(s,p,q)$. Further, since $p\in \overline{uv}\subseteq \pi(s,r)$, $\alpha$ is in $\triangle(s,r,q)$ and $q$ is an apex of $\triangle(s,r,q)$.

This proves the lemma statements (1) and (2). The lemma statements (3) and (4) also immediately follow.

Finally, we argue that $\overline{\alpha p_t}$ intersects $\overline{uv}$.
Assume to the contrary that this is not true. Since $\alpha\in P_1$ and $t\in P_2$, $\pi(\alpha,t)$ must cross $\overline{uv}$ at a point $z$. As $\overline{\alpha p_t}$ does not intersect $\overline{uv}$, $p_t$ must be a polygon vertex in $\pi(\alpha,z)$, which is subpath of $\pi(\alpha,t)$. As $z\in \overline{uv}\subseteq\pi(s,r)$, $\pi(\alpha,z)$ must be ``between'' $\pi(\alpha,s)$ and $\pi(\alpha,r)$. Since no two paths of $\pi(\alpha,s)$, $\pi(\alpha,z)$, and $\pi(\alpha,r)$ cross each other and $P$ is a simple polygon, $p_t$ must be in either $\pi(\alpha,s)$ and $\pi(\alpha,r)$. As $p_t\in \pi(\alpha,t)$ and $\alpha$ is equidistant from $s$, $t$, and $r$, $p_t$ is on the bisector between $t$ and one of $s$ and $r$. This contradicts with our general position assumption since $p_t$ is a vertex of $P$.

The above proves that $\overline{\alpha p_t}$ intersects $\overline{uv}$, i.e., the lemma statement~(6), which also leads to the lemma statement~(5).
\end{proof}

Due to the preceding lemma, our algorithm works as follows. First, we compute the vertices $u'$, $v'$, and $t_{uv}$, which can be done in $O(\log n)$ time by the GH data structure. Then we apply the tentative prune-and-search technique~\cite{ref:KirkpatrickTe95} on the three pseudo-convex chains specified in the lemma in a similar way as before to compute $\alpha$ in $O(\log n)$ time. Finally, we validate $\alpha$ in $O(\log n)$ time in a similar way as before. The overall time of the algorithm is $O(\log n)$.
Lemma~\ref{lem:10} is thus proved.

\paragraph{Remark.} As discussed above, there are two mistakes in the algorithm of Oh and Ahn~\cite{ref:OhVo20}  (Lemma 3.6): (1) In the subcase $c\in \pi(s',r')$, the case where not both $u$ and $v$ are polygon vertices is missed; (2) in the subcase $c\not\in \pi(s',r')$, they erroneously claimed that $\alpha$ does not exist. Both mistakes can be corrected with our new results. Indeed, in both cases we have proved that $\overline{\alpha p_t}$ intersects $\overline{uv}$. With this critical property, their algorithm of Lemma 3.6~\cite{ref:OhVo20} (which was originally designed for the case where $c\in \pi(s',r')$ and both $u$ and $v$ are polygon vertices) can be applied to compute $\alpha$ in $O(\log^2 n)$ time.
In this way, Lemma~3.6 of~\cite{ref:OhVo20} is remedied and thus all other results of~\cite{ref:OhVo20} that rely on Lemma~3.6 are not affected.

\footnotesize
 \bibliographystyle{plain}
\bibliography{reference}

\end{document}